%% file: main_Min-Cut_latex.tex
\documentclass[11pt]{article}

\usepackage{geometry}
\usepackage{tikz}
\usetikzlibrary{backgrounds,positioning,fit,patterns,shadows,calc}
\usepackage[normalem]{ulem} %% For strikingo
\usepackage{epsfig}
\usepackage{amsfonts}
\usepackage{amssymb}
\usepackage{amstext}
\usepackage{amsmath}
\usepackage{datetime}
\usepackage{enumerate}
\usepackage{calligra}
\usepackage[T1]{fontenc}
\usepackage{mathtools}
\usepackage{amsthm, thmtools}
\usepackage{paralist}
\usepackage{nameref}
\definecolor{ForestGreen}{rgb}{0.1333,0.5451,0.1333}
\definecolor{DarkRed}{rgb}{0.8,0,0}
\definecolor{Red}{rgb}{1,0,0}
\usepackage[linktocpage=true,
pagebackref=true,colorlinks,
linkcolor=DarkRed,citecolor=ForestGreen,
bookmarks,bookmarksopen,bookmarksnumbered]
{hyperref}
\renewcommand*\backref[1]{\ifx#1\relax \else (cit. on p. #1) \fi} %http://latex.org/forum/viewtopic.php?t=3670
\usepackage{cleveref}
\usepackage{thm-restate} % See section 1.4 of the pdf above
\usepackage{xspace}
\usepackage{color}
\usepackage[lined,boxed,linesnumbered]{algorithm2e}
\usepackage{enumitem}
\usepackage{comment}
\usepackage{caption}

\newcommand{\curly}[1]{\left\{ #1 \right\}}
\newcommand{\paren}[1]{\left( #1 \right)}

\newcommand{\angularbrac}[1]{\langle{#1}\rangle}

\newcommand{\CONGEST}{\ensuremath{\mathsf{CONGEST}}\xspace}
\newcommand{\PRAM}{\ensuremath{\mathsf{PRAM}}\xspace}
\newcommand{\poly}{\mbox{\rm poly}}

\newcommand{\myvol}{\operatorname{vol}}
\newcommand{\MSGC}{\operatorname{MSGC}}
\newcommand{\tripartition}{\texttt{Tripartition}}
\newcommand{\id}{\ensuremath{{id}}\xspace}
\newcommand{\clustid}{\ensuremath{{groupId}}\xspace}
\newcommand{\clr}{\operatorname{color}}

\newcommand{\core}{\operatorname{Core}}
\newcommand{\regnodes}{\operatorname{Regular}}
\newcommand{\dist}{\operatorname{distance}}

\newcommand{\tree}{\mathcal{T}}
\DeclareMathOperator{\anc}{anc}
\DeclareMathOperator{\children}{children}
\newcommand\ancestor[2][]{%
	\ifstrempty{#1}{%
		\anc\paren{#2}
	}{%
		{\anc}_{#1}\paren{#2}
	}%
}
\newcommand\level[2][]{%
	\ifstrempty{#1}{%
		\texttt{level}\paren{#2}
	}{%
		\texttt{level}_{#1}\paren{#2}
	}%
}
\newcommand\parent[2][]{%
	\ifstrempty{#1}{%
		\pi\paren{#2}
	}{%
		\pi_{{#1}}\paren{#2}
	}%
}
\newcommand\edge[2][]{%
	\ifstrempty{#1}{%
		e\paren{#2}
	}{%
		e_{{#1}}\paren{#2}
	}%
}

\newcommand\child[2][]{%
	\ifstrempty{#1}{%
		\children\paren{#2}
	}{%
		\children_{{#1}}\paren{#2}
	}%
}

\newcommand\desc[2][]{%
	\ifstrempty{#1}{%
		{#2}^{\downarrow}
	}{%
		{#2}^{\downarrow {#1}}
	}%
}

\def\degree{\mathrm{deg}}
\def\vlow{V_{\operatorname{low}}}

\def\ndel{n^{\gamma}}
\def\ndell{n^{2\delta}}

\def\Remove#1{{\sf Remove}\text{-}#1}
\def\Split#1{{\sf Split}\text{-}#1}
\newcommand{\newvol}{V_{\operatorname{vol}}}

\newcommand{\EhExp}{E_h}
\newcommand{\msg}{\texttt{msg}}

\DeclareMathOperator*{\Prob}{\ensuremath{\textnormal{Pr}}}
\renewcommand{\Pr}{\Prob}

\newcommand{\mix}{\ensuremath{\tau_{\operatorname{mix}}}}
\newcommand{\ID}{{\operatorname{ID}}}
\DeclarePairedDelimiter\ceil{\lceil}{\rceil}

\DeclareMathOperator{\polylog}{polylog}
\DeclareMathOperator{\mapping}{\texttt{mapping} }
\DeclareMathOperator{\false}{false}
\DeclareMathOperator{\true}{true}

\ifdefined\ShowComment

\def\danupon#1{\marginpar{$\leftarrow$\fbox{dn}}\footnote{$\Rightarrow$~{\sffamily #1 --Danupon}}}
\def\thatchaphol#1{\marginpar{$\leftarrow$\fbox{ts}}\footnote{$\Rightarrow$~{\sffamily #1 --Thatchaphol}}}
\def\mohit#1{\marginpar{$\leftarrow$\fbox{md}}\footnote{$\Rightarrow$~{\sffamily #1 --Mohit}}}
\def\monika#1{\marginpar{$\leftarrow$\fbox{mh}}\footnote{$\Rightarrow$~{\sffamily #1 --Monika}}}
\else

\def\danupon#1{}
\def\thatchaphol#1{}
\def\mohit#1{}
\def\monika#1{}
\fi

\makeatletter
\def\thmt@refnamewithcomma #1#2#3,#4,#5\@nil{%
	\@xa\def\csname\thmt@envname #1utorefname\endcsname{#3}%
	\ifcsname #2refname\endcsname
	\csname #2refname\expandafter\endcsname\expandafter{\thmt@envname}{#3}{#4}%
	\fi
}
\makeatother

\declaretheorem[numberwithin=section,refname={Theorem,Theorems},Refname={Theorem,Theorems},name={Theorem}]{thm}
\declaretheorem[numberlike=thm,refname={Lemma,Lemmas},Refname={Lemma,Lemmas},name={Lemma}]{lem}
\declaretheorem[numberlike=thm,refname={Lemma,Lemmas},Refname={Lemma,Lemmas},name={Lemma}]{lemma}
\declaretheorem[numberlike=thm,refname={Observation,Observations},Refname={Observation,Observations},name={Observation}]{obs}
\declaretheorem[numberlike=thm,refname={Corollary,Corollaries},Refname={Corollary,Corollaries},name={Corollary}]{cor}

\declaretheorem[numberlike=thm,refname={Observation,Observations},Refname={Observation,Observations}]{observation}

\declaretheorem[numberlike=thm,refname={Definition,Definitions},Refname={Definition,Definitions},name={Definition}]{defn}
\declaretheorem[numberlike=thm,refname={Claim,Claims},Refname={Claim,Claims}]{claim}

\newboolean{short}
\setboolean{short}{false} % set to true if the paper is the short version

\newcommand{\shortOnly}[1]{\ifthenelse{\boolean{short}}{#1}{}}
\newcommand{\longOnly}[1]{\ifthenelse{\boolean{short}}{}{#1}}

\newcommand{\squishlist}{
	\begin{list}{$\bullet$}
		{ \setlength{\itemsep}{0pt}
			\setlength{\parsep}{2pt}
			\setlength{\topsep}{2pt}
			\setlength{\partopsep}{0pt}
			\setlength{\leftmargin}{1.5em}
			\setlength{\labelwidth}{1em}
			\setlength{\labelsep}{0.5em} } }
	\newcommand{\squishend}{
\end{list}  }
\renewcommand{\paragraph}[1]{\medskip\noindent{\bf #1}\xspace}

\def\*#1*\ {}

\makeatletter
\renewcommand{\paragraph}{%
	\@startsection{paragraph}{4}%
	{\z@}{1ex \@plus 1ex \@minus .2ex}{-1em}%
	{\normalfont\normalsize\bfseries}%
}
\makeatother

\AtBeginDocument{%
	\let\ref\Cref
}

\Crefname{algocf}{Algorithm}{Algorithms}
\crefalias{AlgoLine}{line}%

\makeatletter
\let\cref@old@stepcounter\stepcounter
\def\stepcounter#1{%
	\cref@old@stepcounter{#1}%
	\cref@constructprefix{#1}{\cref@result}%
	\@ifundefined{cref@#1@alias}%
	{\def\@tempa{#1}}%
	{\def\@tempa{\csname cref@#1@alias\endcsname}}%
	\protected@edef\cref@currentlabel{%
		[\@tempa][\arabic{#1}][\cref@result]%
		\csname p@#1\endcsname\csname the#1\endcsname}}
\makeatother
\usepackage{authblk}
\author[1]{Mohit Daga}
\author[2]{Monika Henzinger}
\author[1]{Danupon Nanongkai}
\author[3]{Thatchaphol Saranurak\thanks{Work partially done while at KTH Royal Institute of Technology, Sweden.}}

\affil[1]{KTH Royal Institute of Technology, Stockholm, Sweden}
\affil[2]{University of Vienna, Vienna, Austria}
\affil[3]{Toyota Technological Institute, Chicago, USA}

\begin{document}

	\newcommand*\samethanks[1][\value{footnote}]{\footnotemark[#1]}
	
%\title{Distributed Minimum Cut of a Simple Graph in Sublinear Time}
\title{Distributed Edge Connectivity in Sublinear Time}	
%\title{Deterministic Global Minimum Cut of a Simple Graph in Sublinear Time}
%	\title{Sublinear Time Exact Distributed Min-Cut}

%	\date{\today, \currenttime}	
	\date{}
	\pagenumbering{roman}
	\maketitle
	\input{camera_ready_sub_parts/abstract.tex}

	\pagebreak{}
	\tableofcontents{}
	
	\pagebreak{}
	\pagenumbering{arabic}	
	
\input{camera_ready_sub_parts/intro.tex}

\input{camera_ready_sub_parts/prelim.tex}

%F\input{camera_ready_sub_parts/overview.tex}
\input{camera_ready_sub_parts/connect-cert_update.tex}	

\input{camera_ready_sub_parts/contraction-of-graph}
\input{camera_ready_sub_parts/contraction-of-graph_correctness}

\input{camera_ready_sub_parts/contraction-of-graph_run_time}
\input{camera_ready_sub_parts/min-cut-contracted-graph}
\input{camera_ready_sub_parts/warmUp1}

\input{camera_ready_sub_parts/sec6_final}

\input{camera_ready_sub_parts/concl}
\input{camera_ready_sub_parts/open}

\section*{Acknowledgement}  
	\input{camera_ready_sub_parts/ack.tex}	
	
	\bibliographystyle{alpha}
	\bibliography{main_Min-Cut_latex.bib}

	\pagebreak{}
	
	\appendix
	\input{camera_ready_sub_parts/appendix_section3.tex}	
	\input{camera_ready_sub_parts/appendix_section6.tex}
	\input{camera_ready_sub_parts/appendix_exp_decomposition}
	\end{document}

%% file: camera_ready_sub_parts/abstract.tex
% V3:Go very short. No talk about general exact vs. approx
\begin{abstract}
%Consider computing	edge connectivity and finding the minimum-cardinality cut on distributed message-passing networks (the CONGEST model). 
%
%We present the first sublinear-time algorithm that can compute the edge connectivity $\lambda$ of a network {\em exactly} on distributed message-passing networks (the CONGEST model), as long as the network contains no multi-edge. 
We present the first sublinear-time algorithm for a distributed message-passing network sto compute its edge connectivity $\lambda$ {\em exactly} in the CONGEST model, as long as there are no parallel edges. 
Our algorithm takes $\tilde O(n^{1-1/353}D^{1/353}+n^{1-1/706})$ time to compute $\lambda$ and a cut of cardinality $\lambda$ with high probability, where $n$ and $D$ are the number of nodes and the diameter of the network, respectively, and $\tilde O$ hides polylogarithmic factors. This running time is sublinear in $n$ (i.e. $\tilde O(n^{1-\epsilon})$) whenever $D$ is. 
Previous sublinear-time distributed algorithms can solve this problem either (i) exactly only when $\lambda=O(n^{1/8-\epsilon})$ [Thurimella PODC'95; Pritchard, Thurimella, ACM Trans. Algorithms'11; Nanongkai, Su, DISC'14] or (ii)  approximately [Ghaffari, Kuhn, DISC'13; Nanongkai, Su, DISC'14].\footnote{Note that the algorithms of [Ghaffari, Kuhn, DISC'13] and [Nanongkai, Su, DISC'14] can in fact approximate the minimum-weight cut.} 
	
To achieve this we develop and combine several new techniques.\danupon{I'm not sure we really need the first sentence.} 
First, we design the first distributed algorithm that can compute a {\em $k$-edge connectivity certificate} for any $k=O(n^{1-\epsilon})$ in time $\tilde O(\sqrt{nk}+D)$. The previous sublinear-time algorithm can do so only when $k=o(\sqrt{n})$ [Thurimella PODC'95]. 
\danupon{The running time in actually $\tilde O(\sqrt{nk}+D)$ but this is not stated yet.}
%Our algorithm takes $\tilde O(\sqrt{nk}+D)$ time, whic
In fact, our algorithm can be turned into the first parallel algorithm with polylogarithmic depth and near-linear work. Previous near-linear work algorithms are essentially sequential and previous polylogarithmic-depth algorithms require $\Omega(mk)$ work in the worst case (e.g. [Karger, Motwani, STOC'93]). 
\danupon{The theorem is not stated or proven.}
Second, we show that by combining the recent  distributed expander decomposition technique of [Chang, Pettie, Zhang, SODA'19] with  techniques from the sequential deterministic edge connectivity algorithm of [Kawarabayashi, Thorup, STOC'15], we can decompose the network into a sublinear number of clusters with {\em small average diameter} and without any mincut separating a cluster (except the ``trivial'' ones). This leads to a simplification of the Kawarabayashi-Thorup framework (except that we are randomized while they are deterministic). This might make this framework more useful in other models of computation. 
Finally, by extending the tree  packing technique from [Karger STOC'96], we can find the minimum cut in time proportional to the number of components. As a byproduct of this technique, we obtain an $\tilde O(n)$-time algorithm for computing exact minimum cut for \emph{weighted} graphs.
%	
%	- sublinear component when graph has high min-degree
%	
%	- connectivity certificate in sublinear time ... Thurimella can only do for small $\lambda$. Side: PRAM polylogarithmic depth (indep of $\lambda$), near-linear work. Previously near-linear work algorithms require depth depending on $\lambda$. [Cheriyan, Thurimella, STOC'91]
%	
%	- Thorup is pretty sequential  (also Henzinger...) ... Also conceptually simpler 
%	
%	- Karger tree packing [STOC'96]
\end{abstract}

%% file: camera_ready_sub_parts/intro.tex
\section{Introduction\label{sec:intro}}

%\thatchaphol{Our sparse $k$-edge connectivity certificate algorithm is the first $\tilde{O}(m)$ work and $\tilde{O}(1)$ depth algorithms for general $k$. I tried look up the literature and it seems all previous works need $\Omega(mk)$ work. (e.g. https://www.sciencedirect.com/science/article/pii/S0196677400911441) 

Edge connectivity is a fundamental graph-theoretic concept measuring the minimum number of edges to be removed to disconnect a graph $G$. We give a new algorithm for computing this measure in the \CONGEST model of distributed networks.
In this model a network is represented by an unweighted, undirected, connected $n$-node graph $G=(V, E)$.
Nodes represent processors with unique IDs and infinite computational power that initially only know their incident edges. 
They can communicate with each other in {\em rounds}, where in each round each node can send a message of size $O(\log n)$ to each neighbor. 
The goal is for nodes to finish some tasks together in the smallest number of rounds, called {\em time complexity}.
%\footnote{This problem is sometimes referred to as {\em name-independent routing schemes}. See, e.g. \cite{LenzenP_stoc13,LenzenP-podc15}  for discussions and results on another variant called {\em name-dependent routing schemes} which is not considered in this paper.}   
%
%It is usually expressed in terms of $n$ and $D$, where $n$ is the numer of nodes and $\diam$ is the diameter of the network when edge weights are omitted. Throughout we use $\tilde \Theta$, $\tilde O$ and $\tilde \Omega$ to hide polylogarithmic factors in $n$. See \Cref{sec:prelim} for details of the model.
%
The time complexity is usually expressed in terms of $n$ and $D$, the number of nodes and the diameter of the network. Throughout we use $\tilde \Theta$, $\tilde O$ and $\tilde \Omega$ to hide polylogarithmic factors in $n$. (See \Cref{sec:prelim} for details of the model.)

There are two natural objectives for computing a network's edge connectivity. The first is to make every node knows the edge connectivity of the network, denoted by $\lambda$. 
%Another objective is for each node to learn which of its incident edges are in the {\em minimum cut} (in short, {\em mincut}), referring to the set of $\lambda$ edges whose removals disconnect the graph.
The second is to learn about a set $C$ of $\lambda$ edges whose removals disconnect the graph, typically called a {\em mincut}. In this case, it is required that every node knows which of its incident edges are in $C$.\footnote{Readers who are new to distributed computing may wonder whether it is also natural to have a third objective where every node is required to know about {\em all} edges in the mincut. This can be done fairly quickly after we achieve the second objective, i.e. in $O(\min(\lambda,\sqrt{n\lambda})+D)$ rounds \cite{Censor-HillelGK14-decomposition}. 
For this reason, we do not consider this objective here.}
% It is easy to argue that this requires $\Omega(\lambda+D)$ time. Moreover, this can be achieved in $O(\lambda+D)$ time after we achieve the second objective. So, we do not discuss this third objective in the paper.} 
%
Since our results and other results hold for both objectives, we do not distinguish them in the discussion below.

It is typically desired that distributed algorithms run in {\em sublinear time}, meaning that they take $\tilde O(n^{1-\epsilon}+D)$ time for some constant $\epsilon>0$.\footnote{An exception is when a linear-time lower bound can be proved, e.g. for all-pairs shortest paths and diameter (e.g. \cite{BernsteinN19,HuangNS17,Elkin-STOC17,LenzenP_podc13,HolzerW12,FrischknechtHW12,PelegRT12,AbboudCK16}).}. 
%Starting from the dist of Garay, Kutten, and Peleg
%
Such algorithms have been achieved for many problems in the literature, such as minimum spanning tree, single-source shortest paths, and maximum flow \cite{Elkin-STOC17,GhaffariL18,ForsterN18,Censor-HillelGK14-decomposition,BeckerKKL16,Nanongkai-STOC14,NanongkaiS14,GhaffariK13,GhaffariKKLP15,KuttenP98,GarayKP98}. In the context of edge connectivity, the first sublinear-time algorithm, due to Thurimella \cite{Thurimella95}, was for finding if $\lambda=1$ (i.e. finding a {\em cut edge}) and takes $O(\sqrt{n}\log^* n + D)$ time. This running time was improved to $O(D)$ by Pritchard and Thurimella \cite{PritchardT11}, who also presented an algorithm with the same running time for $\lambda=2$ (i.e. they can find a so-called {\em cut pair}). More recently, by adapting Thorup's tree packing \cite{Thorup07}, Nanongkai and Su \cite{NanongkaiS14} presented a $O((\sqrt{n}\log^* n + D)\lambda^4)$-time algorithm, achieving sublinear time for any $\lambda=n^{1/8-\epsilon}$. 

%By random sampling, this algorithm can also be easily extended to a $O((\sqrt{n}\log^* n + D)\epsilon^{-5}\log^3(n))$-time algorithms that can  $(1+\epsilon)$-approximate $\lambda$, for any $\lambda$, as well as the minimum-weight cut, giving an improvement over the previous approximation algorithms by Ghaffari and Kuhn \cite{GhaffariK13}. 

To compute $\lambda$ when $\lambda\geq n^{1/8}$, we are only aware of {\em approximation algorithms}. The state-of-the-art is the $O((\sqrt{n}\log^* n + D)\epsilon^{-5}\log^3(n))$-time $(1+\epsilon)$-approximation algorithm of Nanongkai and Su \cite{NanongkaiS14}, which is an improvement over the previous approximation algorithms by Ghaffari and Kuhn \cite{GhaffariK13}. In fact,  both algorithms can approximate the minimum-weight cut, and the running time of $O((\sqrt{n}\log^* n + D)\epsilon^{-5}\log^3(n))$ matches a lower bound of \cite{DasSarmaHKKNPPW12} up to polylogarithmic factors; this lower bound holds even for $\poly(n)$-approximation algorithms and on unweighted graphs \cite{GhaffariK13} (also see  \cite{ElkinKNP14,KorKP13,Elkin06,PelegR00}).

% (it was also later extended to unweighted graphs \cite{GhaffariK13}). In other words, no significantly better approximation algorithm exists. 

%A more dramatic improvement was made when {\em approximate solutions} are allowed: 
% *** \danupon{I omitted the detailed discussions for lower bounds (commented below) since they are not really relevant.} ***
%
%On the negative side, it was known that $\tilde \Omega(\sqrt{n})$ time is needed even to approximate the minimum-weight cut on unweighted multigraphs on a $O(\log n)$-diameter network, due to  Ghaffari and Kuhn \cite{GhaffariK13} (extending from the lower bound for weighted graphs from \cite{DasSarmaHKKNPPW12}).\footnote{In our distributed setting, unweighted multigraphs are different from weighted graphs in that a multi-edge allows more communication while the edge weights do not affect communication.} This lower bound holds even when a fairly large approximation ratio is allowed. 
%%
%In simple graphs, \cite{GhaffariK13} showed that any $\alpha$-approximation algorithm requires $\tilde \Omega(\sqrt{\frac{n}{\alpha\lambda}})$ time on a $\tilde O(\frac{1}{\lambda}\cdot \sqrt{\frac{n}{\alpha\lambda}})$-diameter network.  

Given that approximating edge connectivity is well-understood, a big open problem that remains is whether we can compute $\lambda$ {\em exactly}. This question in fact reflects a bigger issue in the field of distributed graph algorithms: While there are plenty of sublinear-time {\em approximation} algorithms, many of which are tight, very few sublinear-time {\em exact} algorithms are known. This is the case for, e.g., minimum cut, maximum flow, and maximum matching (e.g. \cite{HenzingerKN-STOC16,BeckerKKL16,Nanongkai-STOC14,NanongkaiS14,GhaffariK13,GhaffariKKLP15,AhmadiKO18-matching}). 
To the best of our knowledge, the only exceptions  are the classic exact algorithms for minimum spanning tree \cite{GarayKP98,KuttenP98} and very recent results on exact single-source and all-pairs shortest paths \cite{Elkin-STOC17,GhaffariL18,ForsterN18,BernsteinN19,HuangNS17}. A fundamental question here is whether other problems also admit sublinear-time exact algorithm, and to what extent such algorithms can be efficient. 
\paragraph{Our Contributions.} 
We present the first sublinear-time algorithm that can compute  $\lambda$ {\em exactly} for any $\lambda$. Our algorithm works on simple graphs, i.e. when the network contains no multi-edge.
\begin{thm}\label{thm:main}
There is a distributed algorithm that, after \sloppy $\tilde O(n^{1-1/353}D^{1/353}+n^{1-1/706})$ time, w.h.p. (i) every node knows the network's edge connectivity $\lambda$, and (ii) there is a cut $C$ of size $\lambda$ such that every node knows which of its incident edges are in $C$.\footnote{We say that an event holds with high probability (w.h.p.) if it holds with probability at least $1-1/n^\epsilon$, where $\epsilon$ is an arbitrarily large constant.}
\end{thm}
As a byproduct of our technique, we also obtain a $O(n \text{ polylog}n)$-round algorithm for computing exact minimum cut in \emph{weighted} graphs (see \Cref{thm:min_cut_o_n_general_graph}).

To achieve \Cref{thm:main}, we develop and combine several new techniques from both distributed and static settings. First, note that we can also assume that we know the approximate value of $\lambda$ from \cite{NanongkaiS14,GhaffariK13}.
More importantly, the previous algorithm of \cite{NanongkaiS14} can already compute $\lambda$ in sublinear time when $\lambda$ is small; so, we can focus on the case where $\lambda$ is large here (say $\lambda=\Omega(n^{c})$ for some constant $c>0$). 
Our algorithm for this case is influenced by the static connectivity algorithm of Kawarabayashi and Thorup (KT) \cite{kawarabayashi2015deterministic}, but we have to make many detours. The idea is as follows. In \cite{kawarabayashi2015deterministic}, it is shown that if a {\em simple} graph $G=(V, E)$ of minimum degree $\delta$ has edge connectivity strictly less than $\delta$, then there is a near-linear-time static algorithm that partitions nodes in $G$ into  $\tilde O(n/\delta)$ many clusters in such a way that no mincut separates a cluster; i.e. for any mincut $C\subseteq E$, every edge in $C$ must have two end-vertices in different clusters. Once this is found, we can apply a fast static algorithm on a graph where each cluster is contracted into one node.
%a new graph $H$ with$ nodes and $\tilde O(m/\delta)$ edges. 
%
Since in our case $\delta\geq \lambda=\Omega(n^{\epsilon})$, the KT algorithm gives hope that we can partition our network into  $\tilde O(n/\delta)=\tilde O(n^{1-\epsilon})$ clusters. Then we maybe able to design a distributed algorithm that takes time near-linear in the number of clusters. There are however several obstacles:
\begin{enumerate}[label=(\roman*)]
	\item \label{item:sparse_cert} The KT algorithm requires to start from a {\em $\lambda$-edge connectivity certificate}, i.e. a subgraph of $O(n\lambda)$ edges with connectivity $\lambda$. However, existing distributed algorithms can compute this only for $\lambda=o(\sqrt{n})$. 
	\item \label{item:sequential} The KT algorithm is highly sequential. For example, it alternatively applies the contraction and trimming steps to the graph several times. 
	%it consists of a sequence of decomposition and contraction steps of the graph. 
	\item \label{item:mincut_contracted} Even if we can get the desired clustering, it is not clear how to compute $\lambda$ in time linear in the number of clusters. In fact, there is even no 
%\mohit{\textbf{some issue with the sentence structure, perhaps this should be good: there does not exist even an}} 
$\tilde O(n)$-time algorithm for computing $\lambda$. 
\end{enumerate}

For the first obstacle, the previous algorithm for computing a $\lambda$-edge connectivity certificate is by Thurimella \cite{Thurimella95}. It takes $O((\sqrt{n}\log^*n+D)\lambda)$, which is too slow when $\lambda$ is large. To get around this obstacle, we design a new distributed algorithm that can compute a $\lambda$-edge connectivity certificate in $\tilde O(\sqrt{n\lambda}+D)$ time. 
The algorithm is fairly intuitive: We randomly partition edges into $c=\lambda/\polylog(n)$ groups. Then we compute an $O(\polylog(n))$-edge connectivity certificate for each group {\em simultaneously}. This is doable in $\tilde O(\sqrt{nc}+D)$ time by fine-tuning parameters of Kutten-Peleg's minimum spanning tree algorithm \cite{KuttenP98} and using the scheduling of \cite{ghaffari2015near}, as discussed in \cite{ghaffari2015near}. 

This algorithm also leads to the first parallel algorithm for computing a $2$-edge connectivity certificate with polylogarithmic depth and near-linear work. To the best of our knowledge, previous near-linear work algorithms are essentially sequential and previous polylogarithmic-depth algorithms require $\Omega(mk)$ work in the worst case (e.g. \cite{KargerM97}). 

For the second obstacle, we first observe that the complex sequential
algorithm for finding clusters in the KT algorithm can be significantly
simplified into a few-step algorithm, if we have a black-box algorithm called
\emph{expander decomposition}. Expander decomposition was introduced
by Kannan et al. \cite{KannanVV00} and is proven to be useful for
devising many fast algorithms \cite{SpielmanT04,OrecchiaV11,OrecchiaSV12,KelnerLOS14,CohenKPPRSV17,ChuGPSSW18}
and also dynamic algorithms \cite{NanongkaiS16,NanongkaiSW17,Wulff-Nilsen16a}. 
With this algorithm, we do not need most of the KT algorithm, except some simple procedures called trimming and shaving, which can be done locally at each node. More importantly, we can avoid the long sequence of contraction and trimming steps (we need to apply these steps only once). 
%More importantly, we can avoid the long sequence of contraction, decomposition, shaving, and trimming steps (we need the shaving and trimming steps only once). 
%
Unfortunately, there is no efficient distributed algorithm for computing the expander decomposition.\footnote{It would be possible to obtain this using the balanced sparse cut  algorithm claimed by Kuhn and Molla \cite{KuhnM15}, but as noted in \cite{chang2018distributed}, the claim is incorrect. We thank Fabian Kuhn for clarifying this issue. After our paper is announced, an efficient distributed algorithm for computing balanced sparse cut is correctly shown in \cite{ChangS19}.} However, we can slightly adjust a very recent algorithm by Chang~et~al.~\cite{chang2018distributed} to obtain a weaker variant of the expander decomposition, which is enough for us.
%to achieved something close enough. 
%We can adjust their algorithm to get a decomposition that we need.  

For the third obstacle, our main insight is the observation that the clusters obtained from the KT algorithm (even after our modification) has {\em low average diameter} ($O(n^c)$ for some small constant $c$). Intuitively, if every cluster has small diameter, then we can run an algorithm on a smaller network where we pretend that each cluster is a node. The fact that clusters have lower average degree is not as good, but it is good enough for our purpose: we can adjust Karger's near-linear-time algorithm \cite{karger2000minimum} to compute $\lambda$ in time near-linear in the number of clusters. 

%To this end, we note that running the algorithm without the clust

%we first observe that we can compute $\lambda$ in $\tilde O(n)$ time by adjusting Karger's near-linear-time algorithm \cite{karger2000minimum}. 

%% file: camera_ready_sub_parts/prelim.tex
\section{Preliminaries\label{sec:prelim}}
\paragraph{Model}
We work in the \CONGEST model \cite{peleg2000distributed}. This is a distributed model for networks which allows synchronous message-passing between any two nodes in the network connected by a direct communication link. The bandwidth is considered to be bounded. Also, the links and nodes are considered to be fault resistant. 
More formally defined, in the \CONGEST model, communication network is modeled as a undirected graph $G = (V, E) $ where each node in $V$ models a processor and each pair of nodes $ \{ u, v \} \in E \subseteq \binom{V}{2} $ is modeled as a link between the processors corresponding to $ u $ and $ v $, respectively.
In the remainder of this paper, we identify vertices, nodes and processors. Also, we use edges for links. In the \CONGEST model, at the beginning each node $v \in V$  has a unique identifier $\id(v)$ of size $O(\log{n})$ (where $ n = |V| $) which is known to node $v$ itself and all its neighbors, i.e., the nodes to which $v$ is connected with a direct communication link. For brevity we will assume that for all node $\id(v) \in [n]$. \footnote{This is a restricted property from general \CONGEST model and can be achieved in $O(D)$ rounds.} In \CONGEST model, message passing between any two nodes connected with direct links occur in synchronous rounds. Lets fix an arbitrary node $v\in V$. At the beginning of each round, node $v$ may send to each of its neighbors a message of size $\Theta (\log{n})$ to all its neighbors. Before the next round begins node $v$ may perform internal computation based on all messages it has received so far and its local knowledge of the network. In the $\CONGEST$ model, the complexity of any algorithm is a measure of the total number of rounds required before the algorithm terminates. The internal computation is not charged.

\paragraph{Notations}
We are given a undirected unweighted simple graph $G = (V,E)$ where $V$ is the vertex set and $E$ is the edge set. We use $n = |V|$ and $m = |E|$. Throughout this paper, we will use $\delta$ to denote the min-degree and $\lambda$ for edge-connectivity of the graph. For $E'\subseteq E$, we use $G[E']$ to be the subgraph of $G$ induced by $E'$. Similarly $G[V']$ for $V' \subseteq V$. Also for any graph $H$, we use $Diam(H) \triangleq \max_{u,v \in V}\dist_H(u,v)$ to denote the diameter of graph $G$. For any vertex $v$, $\deg(v)$ is the degree of the vertex. For a  $U\subset V$, $\myvol(U) = \sum_{u \in U}\deg(u)$. For some subgraph $H$ of $G$ we use $\deg_H(v)$ to denote the degree of vertex $v$ in the subgraph $H$. Lack of subscript implies that the degree is considered with respect to given graph $G$. Similarly, we skip subscript for $\myvol$. A cut (edge cut) is a set of edges $C$, whose deletion from the graph partitions the vertex set $V$ into two connected components $\curly{U,V\setminus U}$. We will represent a cut as an edge set, in which case we say \emph{a cut $C$ of $G$}. At times we will use a partition $\curly{U,T = V\setminus U}$ of vertex set $V$ to represent a cut, then we say \emph{a cut $(U,T)$ of $G$}. For any vertex $v$, we call cuts of the form $(\curly{v},V \setminus \curly{v})$ \emph{trivial}. We use $\partial(U)$ to mean the edges in the cut $(U,V\setminus U)$. For any $U \subset V$, conductance of the cut  $(U,V\setminus {U})$ is defined as $\phi(U) \triangleq \frac{\partial(U)}{\min\curly{\myvol(U),\myvol(V\setminus U)}}$. Further, conductance of a graph $G$ is defined by $\Phi(G) \triangleq \min_{U\subset V} \phi(U)$. For brevity, for any $X \subset V$, we use $\Phi(X)$ instead of $\Phi(G[X])$ to mean the conductance of the subgraph $G[X]$.

\paragraph{Organization of this paper}
In \shortOnly{Algorithm }\ref{algo:overview}, we give a high level overview of the min-cut algorithm. In \cref{sec:connect-cert}, we find the Sparse Connectivity certificate.  In \cref{sec:contraction-of-graph}, we give details of our contraction algorithm which guarantees sublinear number of nodes. Lastly, in \cref{sec:min-cut-contracted-graph}, we give details of our algorithm that finds min-cut in the contracted graph.
\begin{algorithm}[h]
	\DontPrintSemicolon
	\SetKwFor{Forp}{for each}{do \normalfont{(\cref{sec:contraction-of-graph})}}{endfor}
	$G \leftarrow \texttt{Sparse-Connectivity-Certificate}$ (\cref{sec:connect-cert})\;
	$\curly{\EhExp,E_r,E_s} \leftarrow \tripartition$ (\cref{thm:fast_exp_decomp})\;
	$\mathcal{X} \leftarrow $ connected components of subgraph $G[\EhExp]$ (High Expansion Components)\;
	\Forp{$X \in \mathcal{X}$}{ 
		$\curly{\core(X),\regnodes(X)} \gets \texttt{TrimAndShave}(X)$\;
		update $G$ by collapsing $\core(X)$ into a single node
	}
	run distributed algorithm to find min-cut in updated graph $G$ (\cref{sec:min-cut-contracted-graph})
	\caption{High Level Overview of Min-Cut Algorithm}
	\label{algo:overview}
\end{algorithm}
\paragraph{Previously known result for finding Min-Cut} In \cref{sec:intro}, we briefly discussed the result from \cite{NanongkaiS14}. Here we state their main result.
\begin{thm}[From \cite{NanongkaiS14}]
	There exists an algorithm in the \sloppy\CONGEST model which finds $1+\epsilon$ approximation of the min-cut in $O((\sqrt{n}\log^*n + D)\epsilon^{-5}\log^3 n)$ rounds where $\epsilon > 0$. Further, exact value of min-cut can be found exactly in $O((\sqrt{n}\log^* n + D)\lambda^4\log^2 n)$ rounds where $\lambda$ is the size of the min-cut.
	\label{thm:NS14_main}
\end{thm}
In this paper, we use \cref{thm:NS14_main} to find the approximate value of min-cut value. This is used in finding the connectivity certificate in \Cref{sec:connect-cert}. Further, in \Cref{sec:concl}, we use the exact version of the algorithm but only limited to restricted values of $\lambda$.

%% file: camera_ready_sub_parts/connect-cert_update.tex
\section{Connectivity Certificate\label{sec:connect-cert}}
In this section, we give our algorithm for $k$-edge connectivity certificate which significantly reduces the number of edges in the graph. In  the resultant sparse connectivity certificate, we sample $O(k n)$ edges from the graph and prove that these edges are enough to guarantee $k$ edge connectivity of the graph.

\begin{thm}
Let $G = (V,E)$ be an unweighted graph and $k\leq \lambda$. Then in total of $\tilde{O}(\sqrt{nk} + D)$ rounds in the \CONGEST model we can find a $k$-edge connectivity certificate $E'$ of size $O(kn)$ such that every vertex $v$ knows all the adjacent edges in $E'$ and w.h.p. for every cut $C$ of $G$ we have $E' \cap C \geq k$.
\label{thm:sparse_certificate_exists}
\end{thm}

The key idea of our sparse connectivity certificate algorithm is as follows: we first pick a set of random ``skeletons'' based on  \cite{karger1999random}. We then construct a small set of spanning forests in each random skeleton. Further, we argue that the union of all spanning forests leads to the required connectivity certificate.

\begin{thm}
\label{thm:random_skeletons_all_cuts}
Let $G = (V,E)$ be any unweighted, undirected graph. Let $E' \subset E$ be such that each edge $e \in E$ is independently included in $E'$ with probability $p = O\paren{\frac{\ln n}{\epsilon^2 k}}$ for any $k \leq \lambda$. Then w.h.p. all cuts $C$ have less than $(1+\epsilon)p|C|$ edges sampled in $E'$. 
\end{thm}

\cref{thm:random_skeletons_all_cuts} is a standard argument relating random sampling and the proof 
\longOnly{is left to \cref{appendix:proof_random_skeletons_all_cuts} for completeness.}
\shortOnly{is given in the full version of this paper \cite{distributed-min-cut-full}.}
In \shortOnly{Algorithm }\ref{algo:connect-cert}, we give the sequential version of our distributed algorithm to find sparse connectivity certificate.
Further, in \cref{claim:correctness_sparse_cert_tec}, we prove that for every cut $C$ of $G$, w.h.p., at least $k$ edges finally make to the $k$-edge connectivity certificate which implies  that the edge connectivity is at least $k$ as shown in \cref{cor:correctness}. 
Lastly, in \cref{lemma:number_of_edges_Sparse_Connectiviity} we show that the number of edges selected in the  $k$-edge connectivity certificate is $O(k n)$ thus completing the correctness argument. 

\begin{algorithm}
	\SetEndCharOfAlgoLine{}
	\SetKwInOut{Output}{output}
	\Output{$E'$ is $k$-edge connectivity certificate}
	fix $p =  \tau \ln n/(\epsilon^2 k)$, where $\tau$ is some constant such that $1/p$ is an integer for any $\epsilon \in (0,1)$. \label{line:2}\;
	Give each edge a random color in $\curly{1, 2, ..., c}$, where $c=1/p$. \label{line:split_edge_set}\;
	Partition the edge set $E$ into $\mathcal{E} = \curly{E_1,\ldots,E_c}$ such that $E_i = \curly{e \in E \mid \clr\text{ of }e\text{ is }i}$ for all $i \in [c]$\;
	Let $H_1,\ldots,H_c$  be the subgraphs induced by above edge sets.\label{line:set_of_subgraph}\;
	\tcc{In each subgraph $H_i$ construct a set of $\ceil{\frac{(1+\epsilon)\tau \ln n}{\epsilon^2}}$ spanning forests as follows:}
	\For{$i \in [c]$}{\label{line:sparse_cert_pack_spanning_forests}
			${E}_i' \leftarrow \emptyset$\;
		\For{$j=$ $1$ to $\ceil{(1+\epsilon)\frac{\tau \ln n}{\epsilon^2}}$}{
			$E_{i}^j \leftarrow$ edges in an arbitrary spanning forest constructed in the subgraph $H_i$\;
			${E}_i' = {E}_i'\cup E_{i}^j$\;
			$H_i \leftarrow H_i \setminus E_{i}^j$ \tcp{remove edges $E_{i}^j$ from the subgraph $H_i$}
		}	
	}
	\Return $E' = \bigcup_i E_i'$ \label{line:return_line}
	\caption{$\texttt{k-edge connectivity certificate}(G,k)$}
	\label{algo:connect-cert}
\end{algorithm}
\begin{lem}
For any $\epsilon \in (0,1)$, let $k \leq (1+\epsilon)\lambda$, let $E'$ be the  $k$-edge connectivity certificate returned by \shortOnly{Algorithm }\ref{algo:connect-cert}. For any cut $C$, let $F(C) = C \cap E'$.  Then for all cuts $C$ of $G$ we have $|F(C)| \geq \min\paren{k,|C|}$ w.h.p.
\label{claim:correctness_sparse_cert_tec}
\end{lem}
\longOnly{
\begin{proof}
Let $H_1,H_2,\ldots,H_c$ be the set of subgraphs in  \cref{line:set_of_subgraph} of \shortOnly{Algorithm }\ref{algo:connect-cert}. For any cut $C$ of $G$, let $C_i = H_i \cap C$. These are the set of edges sampled in the subgraph $H_i$ from the cut $C$ in \shortOnly{Algorithm }\ref{algo:connect-cert}. Using the value of $p$ chosen in \shortOnly{Algorithm }\ref{algo:connect-cert} and \cref{thm:random_skeletons_all_cuts}, we know that w.h.p. for all cuts $C$ of $G$ we have
	
\begin{equation}
\label{eqn:correctness_sparse_certificate_1}
|C_i| \leq (1+\epsilon)\frac{\tau \ln n}{\epsilon^2 k} |C|
\end{equation}  
	
In this claim, sampling of edges in a subgraph $H_i$ is the only randomized part. Let's fix an arbitrary cut $C$. Since $k \leq \lambda$ thus $|C| \geq k$. In our algorithm, we construct $\ceil{\frac{(1+\epsilon)\tau \ln n}{\epsilon^2}}$ spanning forests one after the other in each subgraph $H_i$ and aggregate the edges of all these spanning forests in $E_i'$. 
Hence, when $|C_i| \geq \ceil{\frac{(1+\epsilon)\tau \ln n}{\epsilon^2}}$, at least $\ceil{\frac{(1+\epsilon)\tau \ln n}{\epsilon^2}}$ of the total edges from $C_i$ make it to $E_i'$. Let $F_i(C) = C \cap E_i'$. 
These are the edges from $C$ in the subgraph $H_i$ which are finally selected to the sparse connectivity certificate. 
We segregate the subgraphs into two sets based on $|C_i|$. Let $B$ be the set of indices corresponding to the subgraphs such that $|C_i| \geq \ceil{\frac{(1+\epsilon)\tau \ln n}{\epsilon^2}}$. Recall that $C_i$ is the set of edges sampled in the subgraph $H_i$ from the cut $C$. Also, each edge of $C$ belongs to exactly one subgraph $H_i$ by  \cref{line:set_of_subgraph} of \shortOnly{Algorithm }\ref{algo:connect-cert}. Thus, we have

\begin{equation}
\sum_{i \in B}|C_i| = |C|- \sum_{i \notin B}|C_i| = |C| - \sum_{i \notin B}F_i(C)
\label{eqn:correctness_sparse_certificate_5}
\end{equation}
Here the last equality is true because when $|C_i| \leq \ceil{(1+\epsilon)\frac{\tau \ln n}{\epsilon^2}}$ then all the edges from $C_i$ make to $E'$, thus $F_i(C) = C_i$. Also $F(C) = \bigcup_i F_i(C)$ thus,
	
\begin{align}
|F(C)| &= \sum_i |F_i(C)| \nonumber\\
& \geq \sum_{i \in B} {\frac{(1+\epsilon)\tau \ln n}{\epsilon^2}} + \sum_{i \notin B} |F_i(C)|&\text{\ let } x= \sum_{i\notin B} |F_i(C)|\nonumber\\
&\geq \frac{(1+\epsilon)\tau \ln n}{\epsilon^2}\sum_{i\in B} \frac{|C_i|}{(1+\epsilon)\frac{\tau \ln n}{\epsilon^2 k} |C|} + x&\text{using \cref{eqn:correctness_sparse_certificate_1}} \nonumber\\
&= \frac{(1+\epsilon)\tau\ln n}{\epsilon^2}\cdot\frac{|C|-x}{(1+\epsilon)\frac{\tau \ln n}{\epsilon^2 k} |C|} + x& \text{by \cref{eqn:correctness_sparse_certificate_5}} \nonumber\\
&=\frac{|C|-x}{\frac{|C|}{k}} + x \nonumber\\
\label{eqn:correctness_sparse_certificate_6}	
\end{align}
If $|C| > k$, then $\frac{|C|-x}{\frac{|C|}{k}} + x > k-x +x = k$ and if $|C| \leq k$ then $\frac{|C|-x}{\frac{|C|}{k}} + x \geq |C|-x +x = |C|$. Also this is true w.h.p for all cuts $C$. Because \cref{eqn:correctness_sparse_certificate_1} holds for all cuts w.h.p.
\end{proof}
}

\begin{cor}
If $k \leq (1+\epsilon)\lambda$, then the $k$-edge connectivity certificate output by \shortOnly{Algorithm }\ref{algo:connect-cert} has edge connectivity at least $\min(k,\lambda)$ w.h.p. 
\label{cor:correctness}
\end{cor}
\begin{proof}
	For each $C$ of $G$, any edge in $C$ is selected at most once in the edge connectivity certificate output by \shortOnly{Algorithm }\ref{algo:connect-cert}. Also, w.h.p., for all cuts $C$ of the graph $G$, from \cref{claim:correctness_sparse_cert_tec}, if $|C| > k$, at least $k$ edges are included in the $k$-edge connectivity certificate and if $|C|\leq k$ then all the edges from cut $C$ are included in the  $k$-edge connectivity certificate. Hence w.h.p. the edge connectivity of $E'$ is  $\min(k,\lambda)$.
\end{proof}
\begin{lem}
	\label{lemma:number_of_edges_Sparse_Connectiviity}
	The number of edges in the  $k$-edge connectivity certificate returned by \shortOnly{Algorithm }\ref{algo:connect-cert} is $O(k n)$.
\end{lem}
\longOnly{
\begin{proof}
	In \shortOnly{Algorithm }\ref{algo:connect-cert}, we partition edge set into $c = \frac{\epsilon^2 k}{\tau \ln n}$ subsets. Further, in each partition we construct $\ceil{\frac{(1+\epsilon)\tau \ln n}{\epsilon^2}}$ many spanning forest and use them as sparse connectivity certificate. We know that a spanning forest has at most $n-1$ edges. Thus in total we have $O(kn)$ edges in the  $k$-edge connectivity certificate.
\end{proof}
}
\shortOnly{
\subsection{Distributed and Parallel Implementation of Algorithm \ref{algo:connect-cert}}
}
\longOnly{\subsection{Distributed and Parallel Implementation of \ref{algo:connect-cert}}}
We now show how to implement \shortOnly{Algorithm }\ref{algo:connect-cert} in the $\CONGEST$ and $\PRAM$ model. We  start with the  $\CONGEST$  model. We give the required algorithm in  \shortOnly{Algorithm }\ref{algo:dist_sparse_connectivity_certificate}. This is a two phase algorithm. In the first phase we randomly partition the edge set and in the second phase construct a connectivity certificate in each partition. For constructing the connectivity certificate, we use a known result about finding a \emph{$l$-slot MST}. The $l$-slot version of the MST problem is as follows: for a given  graph $G=(V,E)$ and $l$ weight functions $W_1$ to $W_l$, where $W_i : E \rightarrow \mathbb{R}$; the $l$-slot MST problem is to find an MST for each of the weight function $W_i$ for $1\leq i \leq l$. The following theorem about $l$-slot MST is obtained by fine-tuning parameters of Kutten-Peleg's minimum spanning tree algorithm \cite{KuttenP98} and using the scheduling of \cite{ghaffari2015near}. We refer to the concluding remarks of \cite{ghaffari2015near}  for details. 
%implicit from \cite{ghaffari2015near}.
\begin{thm}
	The $l$-slot MST problem can be solved in $\tilde{O}(D + \sqrt{nl})$ rounds.
	\label{thm:c-slot-algorithm}
\end{thm}
\begin{algorithm}[h]
 \SetKwInOut{Input}{input}
\SetKwInOut{Output}{output}
\SetKwFor{Forp}{for}{parallely}{endfor}
\SetEndCharOfAlgoLine{}
\SetKwProg{phase}{Phase}{}{}
\SetKwFunction{FMain}{Dist-Sparse-k-Connectivity-Certificate}

\SetKwProg{Fn}{}{}{}
\Fn{\FMain{$G$,$\epsilon$}}{
\Output{$E'$ (edges in  $k$-edge connectivity certificate)\\ when an edge $e\in E'$, both the end points of $e$ know about it}
$E' \leftarrow \emptyset$\;
\phase{1: Partition edge set by assigning random color to each edge}{
A leader node $v$ broadcasts the value of $c = \frac{1}{p}$, where $p = \tau \ln n/(\epsilon^2/k)$ \label{line:broadcast}\;
Construct $W_1$ to $W_c$ weight functions such that $W_i(e) \leftarrow  \infty,\ \forall e\in E,\ \forall i\in [1,c]$\;
\Forp{all node $x \in V$, $\forall e = (x,y) \in E$}{
	\If{$\id(x) > \id(y)$}{
		$\clr(e) \sim \mathcal{U}(1,c)$ \tcc{$\mathcal{U}(1,c):$ uniform random value from $\curly{1,\ldots,c}$}
		send $\clr(e)$ to $y$\;
		$W_{\clr(e)}(e) \leftarrow  1$
	}	

}}
\phase{2: Construct a connectivity certificate in each partition}{
\For{$j \in [1,\ceil{\frac{\tau \ln n}{\epsilon^2}}$}{ \label{line:dist_scc_iteration}
	construct $c$-slot MST with weight functions $W_1$ to $W_c$. \;
	Let $M_i$ be the edges in the MST corresponding to the weight $W_i$.\;
	\For{$i \in [1,c]$, $e \in M_i$}{
		\If{$W_i(e) = 1$}{
			add $e$ to $E'$\;
			$W_i(e) \leftarrow \infty$
		}
	}	
}}
\Return{$E'$}
}
\caption{$\texttt{Dist-Sparse-k-Connectivity-Certificate}$}
\label{algo:dist_sparse_connectivity_certificate}
\end{algorithm}
Using the $l$-slot MST algorithm we give a distributed version of \shortOnly{\shortOnly{Algorithm }}\ref{algo:connect-cert} in \shortOnly{Algorithm }\ref{algo:dist_sparse_connectivity_certificate}. Here we make $c$ disjoint partitions of the edges set $E$. This is done by assigning $c$ weight functions $W_1(e),\ldots,W_c(e)$ to each edge $e$ and if an edge $e$ belongs to some partition $i$ then $W_i(e) = 1$ otherwise $W_i(e) = \infty$. We then construct $\ceil{\frac{(1+\epsilon)\tau \ln n}{\epsilon^2}}$ spanning forests in each of these partitions. This is done by constructing c-slot MST using these weight functions iteratively $\ceil{\frac{(1+\epsilon)\tau \ln n}{\epsilon^2}}$ times. Further, in any iteration $j$, while constructing a spanning forest in a partition $i$, an edge $e$ from constructed MST is selected if it belongs to the partition and has not been selected in a spanning forest prior to iteration $j$. This is ensured by appropriately checking the weight $W_i(e) =1$ and assigning $W_i(e) \leftarrow \infty$.
\begin{lem}
	The  distributed $k$-edge connectivity certificate procedure given in \shortOnly{Algorithm }\ref{algo:dist_sparse_connectivity_certificate} requires $\tilde{O}(D+ \sqrt{nk})$ rounds. \label{lem:sparse_certificate_complexity} 
\end{lem}\longOnly{
\begin{proof}
	In  \shortOnly{Algorithm }\ref{algo:dist_sparse_connectivity_certificate}, we select an arbitrary leader node $v$, which broadcasts the value of $c$ to all nodes. This takes $O(D)$ rounds. To assign a random color to each edge as in \cref{line:split_edge_set}, every vertex $x \in V$, assigns an independently drawn uniformly random $\clr(e) \in [1,c]$ to all the edges $e$ incident on $x$ such that the other end point of $e$ has smaller $\id$ than $x$. If node $x$ assigns $\clr(e)$ to some incident edge $e$ then it communicates the same to the other endpoint of $e$. This takes $O(1)$ rounds.
	Lastly we construct $c$-slot MST for $\ceil{\frac{(1+\epsilon)1\tau \ln n}{\epsilon^2}}$ many times, where $c = \frac{\epsilon^2 k}{\tau \ln n}$. By \cref{thm:c-slot-algorithm}, we know that to computer $c$-slot MST in $\tilde{O}(D+ \sqrt{nc}) = \tilde{O}(D+ \sqrt{nk})$ rounds.
\end{proof}}
To prove the correctness of  \shortOnly{Algorithm }\ref{algo:dist_sparse_connectivity_certificate}, we make the following simple observation.
\begin{obs}
	\label{obs:simple_MST}
	Let $G = (V,E)$ be a simple weighted graph with weight function $W$. Let $E' \subset E$ such that $w(e) = 1\ \forall e\in E'$ and $w(e) = \infty\ \forall e\notin E'$. Let $T$ be an MST of $G$. Construct a forest $T'$ by remove all edges $e$ from $T$ such that $w(e) = \infty$. Then $T'$ is a spanning forest of the subgraph $G[E']$.
\end{obs}
\begin{lem}
	\shortOnly{Algorithm }\ref{algo:dist_sparse_connectivity_certificate} correctly finds the  $k$-edge connectivity certificate.
\end{lem}\longOnly{
\begin{proof}
The edge partition  established in 
\shortOnly{Algorithm }\ref{algo:dist_sparse_connectivity_certificate} is similar to \shortOnly{Algorithm }\ref{algo:connect-cert}. This is because in both  cases each edge is assigned an independent random color from $1$ to $c$ which governs which partition an edge is assigned. Let $\mathcal{E} = \curly{E_1,\ldots,E_c}$ be this partition. To complete this proof we have to show that in both the sequential and the distributed algorithm, the way the spanning forests are constructed in each of the partition is the same.  Let's pick an arbitrary edge set $E_i$ in $\mathcal{E}$ and let $H_i$ be the subgraph induced by $E_i$. In \shortOnly{Algorithm }\ref{algo:connect-cert}, we iteratively construct $\ceil{\frac{\tau \ln n}{\epsilon^2}}$ many spanning forests one after another in the subgraph induced by $H_i$. In \shortOnly{Algorithm }\ref{algo:dist_sparse_connectivity_certificate}, we re-weight the edges used earlier in a spanning forest to $\infty$ and compute an MST. This is done $\ceil{\frac{\tau \ln n}{\epsilon^2}}$ times.  By \cref{obs:simple_MST}, this is similar to constructs spanning forests one after the other in the subgraph $H_i$ resulting in a sparse connectivity certificate.
\end{proof}}
Using \cref{thm:sparse_certificate_exists}, we give the following corollary which will be used in finding the graph contraction in \cref{sec:contraction-of-graph} where we will call this using $\texttt{Sparse-Connectivity-Certificate}(G,\epsilon)$ for some $\epsilon \in (0,\frac{1}{2})$.
\begin{cor}[From \cref{thm:sparse_certificate_exists}]

	Let $0 < \epsilon < \frac{1}{2}$ be a constant. Let $G = (V,E)$ be an unweighted graph with $\lambda$ being the edge connectivity. Then in total of $\tilde{O}(n^{1-\epsilon} + D)$ rounds in the \CONGEST model, we can find a sparse connectivity certificate $E'$ of size $O(\lambda^{\frac{1}{1-2\epsilon}} n)$ such that  the induced subgraph $G[E']$ has connectivity $\lambda$ and every vertex $v$ knows all the adjacent edges in $E'$.
	\label{corr:sparse_certificate_exists}
\end{cor}
\longOnly{
\begin{proof}
	We use \cite{NanongkaiS14}, which for any $\epsilon>0$ finds the $(1+\epsilon)$ approximation of min-cut in $O((\sqrt{n} \log^* n + D)\epsilon^{-5}\log^3 n)$ rounds. Let $\lambda'$ be this approximate value thus $\lambda' \leq (1+\epsilon)\lambda$. If $\lambda' < n^{1-2\epsilon}$, then we output the result of \sloppy$\texttt{Dist-Sparse-k-Connectivity-Certificate}(G,\lambda')$ (see \shortOnly{Algorithm }\ref{algo:dist_sparse_connectivity_certificate}). 
	By \cref{lem:sparse_certificate_complexity}, it takes $\tilde{O}(D + n^{1-\epsilon})$ rounds to compute 
	and by \cref{lemma:number_of_edges_Sparse_Connectiviity}, it has $O(\lambda n)$ edges. 
	Also, based on \cref{cor:correctness}, the edge connectivity is $\lambda$. If $\lambda' \geq n^{1-2\epsilon}$, we output the whole edge set. This is of size $O(n^2) = O(\lambda^{\frac{1}{1-2\epsilon}} n)$ and has the required edge connectivity.
\end{proof}}
\shortOnly{\paragraph{Parallel Implementation of Algorithm \ref{algo:connect-cert}}
}
\longOnly{\paragraph{Parallel Implementation of \ref{algo:connect-cert}}}
In this subsection, we prove that \shortOnly{Algorithm }\ref{algo:connect-cert} has an efficient parallel implementation to find a $\texttt{Sparse-k-Connectivity-Certificate}(G,k)$ taking $\tilde{O}(1)$ depth and total of $\tilde{O}(m)$ work. Recall that in \shortOnly{Algorithm }\ref{algo:connect-cert}, we partition the edge set into $c$ partitions where $c$ depends on $k$. This takes $O(m)$ work and $O(1)$ depth. Now in each partition, we construct $\polylog$ many spanning forests one after the other. This process is done independently in each partition. Also, each edge participates in construction of $\polylog$ many spanning forests thus the total work is $\tilde{O}(m)$.
\begin{thm}
Let $G = (V,E)$ be an unweighted graph, let $k \leq \lambda$. Then in $\tilde{O}(1)$ depth and total of $\tilde{O}(m)$ work in the \PRAM model we can find a $k$-edge connectivity certificate $E'$ of size $O(kn)$ such that every vertex $v$ knows all the adjacent edges in $E'$ and w.h.p. for every cut $C$ of $G$ of size at least $k$ we have $E' \cap C \geq k$.
\label{thm:parallel_sparse_certificate_exists}
\end{thm}

%% file: camera_ready_sub_parts/contraction-of-graph.tex
\section{Graph Contraction\label{sec:contraction-of-graph}} 
In this section, we describe an algorithm which outputs a contracted graph. It uses the sparse connectivity certificate  from \cref{sec:connect-cert} and the graph decomposition from \cref{thm:fast_exp_decomp}.  This contracted graph preserves all {non-trivial} min-cuts and has a sub-linear number of nodes as in {\cite{kawarabayashi2015deterministic} and \cite{henzinger2017local}}. The idea essentially is to pick specialized vertex sets and contract them. Any such contracted vertex set is called as \emph{core}.  We formally define the contraction in the following definition.
\begin{defn}[Min-Cut preserving Sublinear Graph Contraction ($\MSGC(G,\epsilon)$)]
	\label{defn:Min-Cut-Preserving-Sub-linear-Graph-Contraction}
	Let $G = (V,E)$ be a simple unweighted graph such that $m = |E|$ and $n = |V|$. For $0 < \epsilon< 1$, an $\MSGC(G,\epsilon)$ partitions the vertex set $V$	into $\mathcal C = \curly{C_1,C_2,\ldots,C_{O({n}^{1-\Theta(\epsilon)})}}$ and contracts specialized vertex set with the following properties:
	\begin{enumerate}
		\item For all $C \in \mathcal{C}$ , such that $|C| > 1$, we partition $C = \core(C) \cup \regnodes(C)$. $\core(C)$ is called the {\em core} of $C$. $\regnodes(C)$ is a set of {\em regular nodes} in $C$. If $|C| = 1$ (\emph{trivial vertex group}), we set $C = \regnodes(C)$.
		\item For every $C \in \mathcal{C}$, the vertices in $\core(C)$ can be contracted to form a \emph{core vertex} $s(C)$ by deleting the edges which have both end points in $\core(C)$ and collapse the nodes in $\core(C)$ to one node. The contracted graph  thus formed is the $\MSGC(G,\epsilon)$.  Here, $s(C)$ is a vertex of $\MSGC(G,\epsilon)$. 
		\item $\MSGC(G,\epsilon)$ preserves all non-trivial min-cuts of $G$
		\item $\sum_{C \in \mathcal C}diam(G[C]) = O({n}^{1-\Theta(\epsilon)})$
		\item $\sum_{C \in \mathcal C}|\regnodes(C)| = O({n}^{1-\Theta(\epsilon)})$
	\end{enumerate}
	\label{defn:contracted_graph}
\end{defn}

\begin{thm}
	Let $G=(V,E)$ be a given simple unweighted graph. Let $\epsilon \in (0,\frac{1}{2})$ be such that $\delta = n^{2\epsilon}$, then we can find an $\MSGC(G,\epsilon)$  as given in \cref{defn:contracted_graph} in $O(n^{1-\epsilon/44})$ rounds such that
	\begin{enumerate}
		\item A partition $\mathcal C = \curly{C_1,C_2,\ldots,C_{O({n}^{1-\frac{\epsilon}{22}})}}$ of the vertex set $V$	is established where each cluster $C \in \mathcal{C}$ has a unique $\clustid(C) \in [2n]$. Henceforth, $\mathcal C$ is called the \emph{set of vertex groups} of $\MSGC(G,\epsilon)$. Also, $\sum_{C \in \mathcal C}diam(G[C]) = O({n}^{1-\epsilon/20})$ and $\sum_{C \in \mathcal C}|\regnodes(C)| = O({n}^{1-\epsilon/22})$.
		\item Every vertex $v \in V$, knows the $\clustid(C)$ of the vertex group $C$ it is part of. When, $|C| > 1$, then node $v$ also knows if it is part of $\regnodes(C)$ or $\core(C)$ 
	\end{enumerate}
	\label{thm:main_contracted_graph_can_be_found}
\end{thm}
Our definition of graph contraction has similar properties as in \cite{henzinger2017local} and \cite{kawarabayashi2015deterministic} (i.e. sublinear number of nodes and preserving min-cuts). We have an additional property regarding the diameter which enables us to give efficient distributed algorithm to find a min-cut (see \cref{sec:min-cut-contracted-graph}). The algorithm to find graph contraction given by \cref{defn:contracted_graph} is described in \shortOnly{Algorithm }\ref{algo:dist_sparse_connectivity_certificate}. Note that our algorithm runs without the ``outer loop'' of the algorithms in \cite{henzinger2017local} and \cite{kawarabayashi2015deterministic}. Thus by using the expander decomposition algorithm in \cite{SaranurakW19}, also leads to a simplified static algorithm to find min-cuts. For our setting, we find the high-expansion components using a recent result by \cite{chang2018distributed}. We change the parameters from their presentation leading to a modified definition to suit our requirements. This is given in \cref{defn:decomposition}. 
%A proof of 	\cref{thm:fast_exp_decomp} is given in \longOnly{\cref{sec:proof_exp_decomposition} for verification.}\shortOnly{full version of this paper \cite{distributed-min-cut-full}.}

\begin{defn}
	For $0< \rho,\gamma < 1$, $\operatorname{Tripartition}(\gamma,\rho)$ of a simple, unweighted, undirected graph $G = (V,E)$ is a partition of the edge set $E$ to $\mathcal E =  \curly{\EhExp, E_s, E_r}$ satisfying the following:
	\begin{enumerate}
		\item 
		Each connected component induced by $\EhExp$ is such that $\Phi(X) \geq \frac{c}{n^{\rho}}$ for some constant $c>0$. 
		\item $E_s = \bigcup_{v \in V}E_{s,v}$, where each vertex $v$ knows about each edge in $E_{s,v}$, edges in $E_{s,v}$ are viewed as oriented away from $v$ and the sub-graph induced by $E_s$ has arborocity $O(n^{\gamma})$. \mohit{\textbf{rework this line}}\mohit{Done.}
		\item $|E_r|	 = O(m^{1-\rho/2})$ and each edge of $E_r$ has endpoints in different connected components in the subgraph induced by the edge set $\EhExp$.
	\end{enumerate}
	\label{defn:decomposition}
\end{defn}
\begin{thm}
	\label{thm:fast_exp_decomp}
	For $0<\gamma,\rho<1$, in  $O(n^{1-\gamma + 10\rho})$ rounds, we can find the $\operatorname{Tripartition}(\gamma,\rho)$ of a graph $G = (V,E)$ which partitions the edge set $E$ to $\mathcal E =  \curly{\EhExp, E_s, E_r}$ such that every node $v$ knows which of its incident edges belong to $\EhExp,E_s$  and $E_r$.
\end{thm}

We use \cref{thm:fast_exp_decomp} in \shortOnly{Algorithm }\ref{algo:contraction}. In the remaining part of this sub-section we  give an overview of \shortOnly{Algorithm }\ref{algo:contraction}. 
 We fix the value of $\epsilon$ such that $\delta = n^{2\epsilon}$.  In this algorithm, we first find a sparse connectivity certificate (from \cref{sec:connect-cert}). Subsequently, in this section, we  use $G$ to represent the graph with reduced number of edges received from sparse connectivity certificate. We then use \cref{thm:fast_exp_decomp} with $\gamma = \epsilon$ and $\rho = \epsilon/11$ resulting in a tripartition of edge set $E$ into $\EhExp,E_s$ and $E_r$. We use the connected components induced by the edge set $\EhExp$, and  by \cref{defn:decomposition} each of these components has high expansion. We then \emph{trim} each component followed by \emph{shaving}. The process of \emph{trimming} and \emph{shaving} are same as  \cite{kawarabayashi2015deterministic} and described below.

\paragraph{Trimming and Shaving} Given $U \subset V$, to be trimmed, such that all $u \in U$ have same $\clustid$. In \emph{trimming} process, we repeatedly remove any vertex $u \in U$ if it has less than $2\deg_G(u)/5$ neighbours in $U$ until it is not possible to remover a vertex further. We call a set of vertices $U \subseteq V$  {\em trimmed} if all vertices $u \in U$ have at least $2\deg_G(u)/5$ of their neighbours in $U$. Suppose $U$ is a vertex set to be {\em trimmed}, let $U' \subset U$ be the set of vertices which are removed from $U$ in this process. Then the set of edges which are lost during the trimming process of the vertex set $U$ are the edges which have one end point in $U'$ and the other in $U \setminus U'$.  Also, each $u'\in U'$ assigns itself a new distinct $\clustid(u)$. The \emph{trimming} phase is followed by a \emph{shaving} phase in \shortOnly{Algorithm }\ref{algo:trimming}. \emph{Shaving} does not induce a  modification of vertex groups rather partitions each vertex group $C$ into two sets: $\core(C)$ and $\regnodes(C)$. For any vertex group $C$, during \emph{Shaving}, all nodes $v \in C$ are put into $\regnodes(C)$ if at least $\deg_G(v)/2-1$ edges incident on $v$ leave $C$. We call all such nodes  $\emph{shaved}$. The remaining vertices from $\core(C)$. 

\begin{algorithm}
	\DontPrintSemicolon
	\SetKwFor{Forp}{for}{parallely}{endfor}
	$G \leftarrow \texttt{Sparse-Connectivity-Certificate}(G,\epsilon/44)$ (\cref{corr:sparse_certificate_exists})\label{line:connectivity_certificate}\;
	Let $\curly{\EhExp,E_r,E_s}$ be the partition of edge sets $E$ found using $\tripartition(\gamma = \epsilon,\rho = \epsilon/11)$ (\cref{thm:fast_exp_decomp})\;
	$\mathcal{X} \leftarrow $ connected components of subgraph $G[\EhExp]$\label{line:connected_components}\;
	$\mathcal{C} \leftarrow \curly{C \mid X \in \mathcal{X}; C\text{ is vertex set of }X}$\; 
	\Forp{$v \in V$}{
			$C \in \mathcal{C}$ be the cluster such that $v \in C$\;
			$\clustid(v) \leftarrow \max_{u \in C}\id(u)$
	}
	\lForp{each $v \in V$}{
		run $\texttt{trim\_shave(v)}$
	}	
	\caption{Algorithm to find $\MSGC(G,\epsilon)$}
	\label{algo:contraction}
\end{algorithm}

%% file: camera_ready_sub_parts/contraction-of-graph_correctness.tex
\shortOnly{\subsection{Correctness of \ref{algo:contraction} \label{sub-sec:contraction_contraction}}}
\longOnly{\subsection{Correctness of \ref{algo:contraction} \label{sub-sec:contraction_contraction}}}
\mohit{\textbf{shall we change the sub-section title to Correctness of Algo. 4 and get rid of the preceding sentence?}}
Let $\mathcal{C}$ be the set of vertex groups output by \shortOnly{Algorithm }\ref{algo:contraction}. In this subsection,  we first prove that collapsing a $\core(C)$ of a vertex group $C \in \mathcal{C}$ does not affect a non-trivial min-cut. Then we show that the number of nodes which are trimmed and shaved are bounded. Lastly, we will show that the sum total of diameter of subgraph induced by the vertex groups in $\mathcal{C}$ is bounded.
\paragraph{Clusters and preserving non-trivial cuts in contracted graph} Similar to \cite{kawarabayashi2015deterministic}, we call $C \subset V$ a \emph{cluster} if for every  min-cut $(U,T)$ of $G$ both $|C \cap T|>2$ and $|C \cap U| >2$ are not true simultaneously. \shortOnly{Algorithm }\ref{algo:contraction} establishes a partition $\mathcal{C}$ of the vertex $V$ set by assigning a $\clustid(v)$ to each vertex $v$, such that each $C \in \mathcal C$ is given by $C = \curly{v \mid \clustid(v) = i}$ for some $i \in [2n]$.

In \shortOnly{Algorithm }\ref{algo:contraction}, we start with a partition of vertex set $\mathcal{C}$ corresponding to connected components induced by the edge set $\EhExp$ (recieved from $\tripartition(\gamma = \epsilon,\rho = \epsilon/11)$). 
We assign a unique $\clustid$ to each $C \in \mathcal{C}$ which is known to every vertex $v \in C$ and thus the vertex group it is part of. We run the distributed algorithm $\texttt{trim\_shave}$ on each node which \emph{trims} each vertex group $C \in \mathcal{C}$ followed by \emph{shaving}.
In the process $\clustid$ values of some vertices are changed. 
In \ref{claim:technical_shaving_1}, we give a technical claim which uses the property of high expansion of each component of $G[\EhExp]$ and properties of \emph{trimming} and \emph{shaving}. Using this claim, we prove that at the end of \shortOnly{Algorithm }\ref{algo:contraction}, each vertex group established by these $\clustid$'s is a \emph{cluster}. Finally, using the properties of the cluster and \emph{shaving} process, in \ref{lemma:no_min_cut}, we show that the $\core(C)$ of a cluster can be collapsed without affecting any non-trivial min-cut.
\begin{claim}
	Let $(T,U)$ be any arbitrary min-cut of the graph $G$. At the end of \shortOnly{Algorithm }\ref{algo:contraction}, for any $i \in [2n]$ let $C = \curly{v \mid \clustid(v) = i}$ be an arbitrary vertex group. Then both $|C \cap T|\geq \frac{\delta}{100}$ and $|C\cap U| \geq \frac{\delta}{100}$ are not true simultaneously.
	\label{claim:technical_shaving_1}
\end{claim}
\begin{proof}
Recall that in \shortOnly{Algorithm }\ref{algo:contraction}, we use $\tripartition(\gamma = \epsilon,\rho = \epsilon/11)$ resulting in a tripartition of the edge set into $\EhExp,E_r$ and $E_s$. Furthermore each non-trivial cluster $C$ is formed by trimming \monika{and shaving}\mohit{basically vertices which are shaved remain in a cluster, thus \sout{and shaving}} some vertices from a connected component $X$ in the subgraph $G[\EhExp]$. Note that in \shortOnly{Algorithm }\ref{algo:trimming}, we set	 the $\clustid(v) = n + \id(v)$ of every trimmed node $v$. Thus at the end of $\texttt{trim\_shave}$, these nodes form a vertex group of  single node \monika{what does "of single node" mean?}\mohit{all trimmed nodes are trivial clusters of single vertex}. 	
When $|C| = 1$, this claim is trivial. Each trimmed vertex group $C$ is formed by trimming some vertices from a connected component $X$ in the subgraph $G[\EhExp]$. Thus for every vertex group $C$ there is a connected component $X$ in the subgraph $G[\EhExp]$ such that $C \subseteq X$. 
WLOG, assume that $\myvol_{G[X]}(T \cap X) < \myvol_{G[X]}(U \cap X)$, otherwise we use $ \myvol_{G[X]}(U \cap X)$ in the below equation. For the sake of contradiction, we assume that both $|C \cap T|\geq \frac{\delta}{100}$ and $|C\cap U| \geq \frac{\delta}{100}$. \longOnly{
We have
\begin{align*}
\lambda &= |E(T , U )|  \geq |E(T \cap X, U \cap X)|  \\
&\geq  \Phi(X) \cdot \myvol_{G[X]}(T \cap X)	& \text{since $\myvol_{G[X]}(T \cap X) < \myvol_{G[X]}(U \cap X)$}\\
&\geq  \Phi(X) \cdot \myvol_{G[X]}(T \cap C) & \text{since $C \subseteq X$}\\
&\geq  \Phi(X) \cdot\frac{2}{5}\delta\cdot|T \cap C| &\text{$C$ is \emph{trimmed}, $\forall u\in C$, $\frac{2\deg_G(u)}{5}$ edges incident to $u$ are in $C$}.\\
&\geq \frac{c}{n^{\rho}} \cdot \frac{2}{5}\delta  \cdot \delta/100& \text{$\Phi(X) = \Omega\paren{\frac{1}{n^{\rho}}}$ from  \ref{defn:decomposition}}\\
&> \frac{c}{\delta} \cdot \frac{2}{5}\delta  \cdot \delta/100 & \frac{1}{n^{\rho}} = \frac{1}{n^{\epsilon/11}} = \frac{1}{(n^{2\epsilon})^{1/22}} = \frac{1}{\delta^{1/22}}> \frac{1}{\delta}\\	
&> \delta& \text{choosing $c = 1000$}
\end{align*}

The above is a contradiction since the size of min-cut can not be larger than the min-degree. Thus both $|C \cap T|\geq \frac{\delta}{100}$ and $|C\cap U| \geq \frac{\delta}{100}$ are not true simultaneously.
}	\shortOnly{
This leads to $\lambda > \delta$ (for details see \cite{distributed-min-cut-full}) which is a contradiction since the size of min-cut can not be larger than the min-degree. Thus both $|C \cap T|\geq \frac{\delta}{100}$ and $|C\cap U| \geq \frac{\delta}{100}$ are not true simultaneously.
%		\begin{align*}
%	&\lambda = |E(T , U )|\\  
%	&\geq |E(T \cap X, U \cap X)|  \\
%	&\geq  \Phi(X) \cdot \myvol_{G[X]}(T \cap X)\quad \quad \quad  \myvol_{G[X]}(T \cap X) < \myvol_{G[X]}(U \cap X)\\
%	&\geq  \Phi(X) \cdot \myvol_{G[X]}(T \cap C)\quad \quad \quad  C \subseteq X\\
%	&\geq  \Phi(X) \cdot\frac{2}{5}\delta\cdot|T \cap C| \quad \quad \quad \quad  \text{$C$ is \emph{trimmed}}\\
%	&\geq \frac{c}{n^{\rho}} \cdot \frac{2}{5}\delta  \cdot \delta/100 \quad \quad \quad \quad \quad  \text{$\Phi(X) = \Omega\paren{\frac{1}{n^{\rho}}}$ from  \ref{defn:decomposition}}\\
%	&> \frac{c}{\delta} \cdot \frac{2}{5}\delta  \cdot \delta/100 \quad \quad \quad \quad \quad  \frac{1}{n^{\rho}} = \frac{1}{(n^{2\epsilon})^{1/22}} = \frac{1}{\delta^{1/22}}> \frac{1}{\delta}\\	
%	&> \delta \quad \quad \quad \quad \quad \quad \quad \quad \quad \quad\text{choosing $c = 1000$}
%\end{align*}
	}\monika{Is this in the wrong line?}\mohit{Line 5 looks correct to me.}
\end{proof}
\begin{lemma}
	Let $(T,U)$ be any min-cut of the graph $G$. At the end of \shortOnly{Algorithm }\ref{algo:contraction}, for any $i \in [2n]$ let $C = \curly{v \mid \clustid(v) = i}$ be an arbitrary vertex group. Then both $|C \cap T|>2$ and $|C \cap U| >2$ are not true simultaneously.
	\label{claim:critical_obs_about_number_of_nodes_in_one_side_of_cut_intersect_suppernode}
\end{lemma}
\begin{proof}
	We know that the size of min-cut is always smaller or equal to the min-degree. WLOG assume that $|C \cap T| \leq  |C \cap U|$. Thus 
	\begin{align}
	\delta &\geq \lambda = |E(T, U)|\geq|E(C \cap T, C \cap U)| \nonumber\\
	&= \myvol_{G[C]}(C \cap T) - |E(C \cap T, C \cap T)|\nonumber\\  
	&\geq (2/5)\cdot \delta \cdot |C \cap T| - |C \cap T|^2 & \text{$G$ is simple}
	\label{eqn:lemma_technical_shaving_2}
	\end{align}
	Here  \cref{eqn:lemma_technical_shaving_2} is true for $|C \cap T| \leq 2$, but this equation cannot be true for any $|C \cap T|$ between $3$ and $\delta/100$. Thus, if $|C \cap T| > 2$, it must hold that $|C \cap T| > \delta/100$. But then $|C \cap U| \geq |C \cap T|$ implies that both $|C \cap U|$ and $|C \cap T|$ are larger than $\delta/100$, which is not possible by Claim 5.6. It follows that $|C \cap T| \leq 2$. Thus  both $|C \cap T|>2$ and $|C \cap U| >2$ are not true simultaneously.
\end{proof}
\begin{lemma}
	\label{lemma:no_min_cut}
	Let $(T,U)$ be any non-trivial min-cut of the graph $G$. Then for every non-trivial cluster $C$ in cluster set of $\MSGC(G,\epsilon)$ either we have $T \cap \core(C) = \emptyset$ or $U \cap \core(C) = \emptyset$.
\end{lemma}
\begin{proof}
	Fix an arbitrary vertex group $C$ of $\MSGC(G,\epsilon)$. Suppose $\core(C) \neq \emptyset$. For each node $r \in \core(C)$, let us define the indegree  of node $r$ w.r.t to $C$ as the number of edges incident on $r$ which have the other endpoint in $C$ and denote this by $\texttt{indegree}_C(r)$. By the property of shaving if $r \in \core(C)$, then we have $\texttt{indegree}_C(r) > \texttt{degree}(r)/2 + 1$. Let $(T,U)$ (here $T \cup U = V$) be a min-cut. From \ref{claim:critical_obs_about_number_of_nodes_in_one_side_of_cut_intersect_suppernode} we  know that both $|C \cap T|>2$ and $|C \cap U|>2$ are not simultaneously true. Suppose $|C \cap T|>2$, thus $|C \cap U|\leq2$. We prove that $U \cap \core(C) = \emptyset$. For the sake of contradiction let's assume $ u \in U \cap \core(C)$. \mohit{prelim:For $A,B \subset V$ let $E(A,B)$ be the set of edges which have one end point in $A$ and the other in $B$, Also $(A,B)$ is defined as a cut if $A\cup B = V$}.
	\longOnly{
		\begin{align*}
		|E(\curly{u},T)| &\geq |E(\curly{u}, T \cap C)| \\
		&= |E(\curly{u},C\setminus U)|\\
		&=\texttt{indegree}_C(u) - |E(\curly{u},U\cap C)|\\
		&>\texttt{degree}(u)/2 + 1 - 1&(|U\cap C| \leq 2 \text{ and the graph is simple})\\
		&\geq \texttt{degree}(u)/2 
		\end{align*}
	}
	\shortOnly{
		\sloppy	\begin{align*}
		&|E(\curly{u},T)| \geq |E(\curly{u}, T \cap C)| \\
		&\quad=|E(\curly{u},C\setminus U)|\\
		&\quad=\texttt{indegree}_C(u) - |E(\curly{u},U\cap C)|\\
		&\quad>\texttt{degree}(u)/2 + 1 - 1\quad \quad (|U\cap C| \leq 2 \text{ \& simple Graph})\\
		&\quad\geq \texttt{degree}(u)/2 
		\end{align*}
	}
	The contradiction comes from the fact that flipping $u$'s side in the cut $(T,U)$ decreases the size of $(T,U)$, in contradiction to the fact that it is a min-cut.
\end{proof}

\paragraph{Number of trimmed and shaved nodes is bounded}
Here, we first prove that the number of edges going between the connected components of the subgraph $G[E_m]$ is bounded. Then by using counting argument, we show that the number of trimmed and shaved nodes is bounded. This is similar to \cite[Lemma 17]{kawarabayashi2015deterministic}.
\begin{claim}
	Let $\epsilon \in (0,1/2)$ be such that $n^{2\epsilon} = \delta$. Let $\mathcal{X}$ be the connected components in the subgraph $G[E_m]$ as in \ref{line:connected_components} of \shortOnly{Algorithm }\ref{algo:contraction} while finding $\MSGC(G,\epsilon)$. Then the total number of edges going between any $X,Y \in \mathcal{X}$ is $O(\delta n^{1-\epsilon/22})$.
	\label{claim:inter_component_edges}
\end{claim}
\begin{proof}
	From \ref{defn:decomposition}, we know that the total number of edges which connect any two components is contributed by $E_r$ and $E_s$. Since the arboricity of sub-graph induced by $E_s$ is $O(n^{\gamma})$ thus we have $|E_s| = O( n\times n^{\gamma})$. Further, the number of edges in $E_r = O(m^{1-\rho/2})$. 
	Thus the total number of edges going between any two components $X,Y$ is $O(n^{1+\gamma} + m^{1-\rho/2})$.  
	In \ref{line:connectivity_certificate}, we have used the sparse connectivity certificate, thus by \ref{corr:sparse_certificate_exists}, we have $m = \lambda^{\frac{1}{1-2\epsilon/44}}n$. 
	Recall that $\rho = \epsilon/11$ and $\delta \geq \lambda$. 
	The total number of trimmed edges is $O(m^{1-\rho/2} + n\cdot n^{\gamma}) = O((\lambda^{\frac{1-\rho/2}{1-\epsilon/22}}n^{1-\rho/2} 
	+ n^{1-\gamma}\cdot n^{2\gamma}) = O(\delta n^{1-\rho/2} + \delta n^{1-\gamma}) = O(\delta \cdot n^{1-\epsilon/22})$.
\end{proof}
\begin{lemma}
	Let $\epsilon \in (0,1/2)$ be such that $n^{2\epsilon} = \delta$. The number of nodes trimmed in \shortOnly{Algorithm }\ref{algo:contraction} to find $\MSGC(G,\epsilon)$  is $O(n^{1-{\epsilon/22}})$.
	\label{lemma:number_trimmed_nodes}
	\monika{\sout{find $\MSGC(G,\epsilon)$}}\mohit{My  apprehension: there is a dependence on $\epsilon$ in the result of this lemma, }
\end{lemma}
\begin{proof}
	By \ref{claim:inter_component_edges}, the number of edges going between the connected components $\mathcal{X}$ at \ref{line:connected_components} of \shortOnly{Algorithm }\ref{algo:connect-cert} is $O(n^{1-\epsilon/22})$. Whenever a node $v$ is \emph{trimmed} from a component $X \in \mathcal{X}$ then this splits the component $X$ into $X\setminus \curly{v}$ and $\curly{v}$ thus updating the component set $\mathcal{X}$. For brevity, let,s say that there are total $c$ edges which go between any two components of $\mathcal{X}$ before the start of \emph{trimming} process. When a node $v$ decides to \emph{trim} from some component $X$, then based on the properties it uses at least $3\deg_G(v)/5$ edges among the $c$ edges going between components. Also node $v$ has at most $2\deg_G(v)/5$ edges going to the vertices in $X$, thus when $v$ is trimmed these are added back to the inter component edges. Hence, trimming a node uses at least $\delta/5$ inter component edges. Thus, there are $O(n^{1-{\epsilon/22}})$ nodes which can be trimmed. By definition \emph{trimmed edges} are the edges which go between connected components of $\mathcal{X}$.
\end{proof}
Using similar argument in the above lemma we can prove that the number of nodes which are \emph{shaved} (removed from a cluster) is bounded by $O(n^{1-\epsilon/22})$. This implies the following lemma. 
\begin{lemma}
	Let $\epsilon \in (0,1/2)$ be such that $n^{2\epsilon} = \delta$. Let $\mathcal{C}$ be the cluster set of $\MSGC(G,\epsilon)$ then  $\sum_{C \in \mathcal C}|\regnodes(C)| = O({n}^{1-\epsilon/22})$.
	\label{lemma:regular_nodes_are_bounded}
\end{lemma}
\paragraph{Aggregate cluster diameter is bounded} Now, we prove that that the aggregate diameter of all clusters output by \shortOnly{Algorithm }\ref{algo:connect-cert} is bounded. This requires us to show that the number of clusters $C$ in $\MSGC(G,\epsilon)$, such that $|C|>1$ is bounded. We prove this in the following lemma.
\begin{lemma}
	Let $\mathcal{C}$ be the vertex groups of $\MSGC(G,\epsilon)$ output by \shortOnly{Algorithm }\ref{algo:connect-cert}. The total number of vertex groups $C \in \mathcal{C}$ such that $|C| > 1$ are $O({n^{1-2\epsilon}})$.	\label{lemma:number_of_clusters}\end{lemma}
\begin{proof}
	As per the definition of {\em trimming}, each node in a non-trivial cluster, at the end of \shortOnly{Algorithm }\ref{algo:contraction} has $\frac{2\delta}{5}$ neighbors in the cluster. Since we are dealing with simple graphs, hence the size of cluster is at least $\frac{2\delta}{5}$. Suppose there are more than $3\frac{n}{\delta}$ non-trivial clusters. Thus, the total nodes in the graph would be $\frac{2\delta}{5}\cdot 3\frac{n}{\delta} = \frac{6}{5}n$, which is a contradiction. Then the number of non-trivial clusters in $\MSGC(G,\epsilon)$ are $O(\frac{n}{\delta}) = O(n^{1-2\epsilon})$.
\end{proof}
\begin{lemma} 	Let $\epsilon \in (0,1/2)$ be such that $n^{2\epsilon} = \delta$. Let $\mathcal{C}$ be the cluster set of $\MSGC(G,\epsilon)$ then $\sum_{C \in \mathcal C}diam(G[C]) = O({n}^{1-\epsilon/20})$.
	\label{lemma:diameter_cluster_sublinear}\monika{When you define vertex groups, also define what you mean by a trivial and a non-trivial vertex group.}\mohit{done, see in \ref{defn:Min-Cut-Preserving-Sub-linear-Graph-Contraction}, do you think that is appropriate place?}
\end{lemma}
\begin{proof}
	We know by \cite{erdHos1989radius}, that a graph of $n$ nodes with min-degree $d$ has a diameter $\ceil{\frac{n}{d}}$. From the \emph{trimming} process, we know that for every non-trivial cluster $C$, each node $v \in C$ has at least $\frac{2}{5}\deg_G(v)$ neighbors in $C$. Thus $diam(G[C]) = \ceil{\frac{C}{\frac{2}{5}\delta}}$. For each trivial cluster $C$, $diam(G[C]) = 1$. Thus
	\longOnly{
	\begin{align*}
	\sum_{C \in \mathcal C}diam(G[C] &= \sum_{\substack{C \in \mathcal C\\|C| > 1}}\ceil{\frac{|C|}{\delta}} + \sum_{\substack{C \in \mathcal C\\|C| = 1}}1\\
	&\leq \sum_{\substack{C \in \mathcal C\\|C| > 1}}\paren{\frac{|C|}{\delta} + 1}  + O(n^{1-\epsilon/20})&\text{trimmed nodes by \ref{lemma:number_trimmed_nodes} are }O(n^{1-\epsilon/20})\\
	&\leq \frac{n}{\delta} +  \sum_{\substack{C \in \mathcal C\\|C| > 1}}1+ O(n^{1-\epsilon/20})& \sum_{C \in \mathcal{C}} |C| \leq n \\																										
	&=O(n^{1-\epsilon/20})
	\end{align*}}
\shortOnly{
	\begin{align*}
		&\sum_{C \in \mathcal C}diam(G[C] = \sum_{\substack{C \in \mathcal C\\|C| > 1}}\ceil{\frac{|C|}{\delta}} + \sum_{\substack{C \in \mathcal C\\|C| = 1}}1\\
		&\quad\leq \sum_{\substack{C \in \mathcal C\\|C| > 1}}\paren{\frac{|C|}{\delta} + 1}  + O(n^{1-\epsilon/20})\\
		&\quad\leq \frac{n}{\delta} +  \sum_{\substack{C \in \mathcal C\\|C| > 1}}1+ O(n^{1-\epsilon/20})& \sum_{C \in \mathcal{C}} |C| \leq n \\																										
		&\quad=O(n^{1-\epsilon/20})
\end{align*}
}
	Here the last equation is true {since $\delta = n^{2\epsilon}$ and since we know by \ref{lemma:number_of_clusters} that the number of non-trivial clusters is $O(n^{1-2\epsilon})$}.
\end{proof}

%% file: camera_ready_sub_parts/contraction-of-graph_run_time.tex
\shortOnly{\subsection{Distributed implementation of Algorithm \ref{algo:contraction} \label{sub-sec:contraction_run_time}}}
\longOnly{\subsection{Distributed implementation of \ref{algo:contraction} \label{sub-sec:contraction_run_time}}}
In \shortOnly{Algorithm }\ref{algo:trimming}, we implement the \emph{trimming} process in the distributed setting. Initially, each vertex $v$ is assigned a $\clustid(v)$ establishing disjoint vertex groups (a partition of vertex set $V$). The algorithm then makes sure that each vertex group is {\em trimmed}. During this process vertices which do not satisfy the {\em trimmed} condition remove themselves from the corresponding vertex group and assign themselves a new distinct $\clustid$. This is followed by \emph{shaving}. In \cref{lemma:number_trimmed_nodes} we prove that the total numbers of nodes which are \emph{trimmed} is bounded. This will allow us to bound the run time of \shortOnly{Algorithm }\ref{algo:trimming}.

\begin{algorithm}
	\DontPrintSemicolon
	\SetKwInOut{Input}{input}
	\SetKwInOut{Output}{output}
	\Input{node $v$ has a $\clustid(v)$ and
		\emph{trimming} and \emph{shaving} are performed on vertex groups.  For any $i$ a vertex group $C_i =\curly{v \mid \clustid(v) = i}$}
	$\texttt{regStatus}(v) \leftarrow \false$\;
	$\mathcal{N}(v) \leftarrow \curly{u \mid (u,v) \in E}$\;
	send $\clustid(v)$ to all $u \in \mathcal{N}(v)$\;
	receive $\clustid(u)$ from all $u \in \mathcal{N}(v)$\;
	\tcc{trimming}
	\While{nodes exist to be trimmed}{
		$GOOD(v) = \curly{u \mid u \in \mathcal{N}(v) \ \text{ and }\ \clustid(v) = \clustid(u) }$\;
		\If{$|GOOD(v)| < \frac{2}{5}\deg_G(v)$}{
			$\clustid(v) \leftarrow n + \id(v)$\;
			send $\clustid(v)$ to all $u \in \mathcal{N}(v)$\;
			{\bf \texttt{break}}
		}
	}
	\tcc{shaving: condition for vertex $v$ to be \lq{shaved}\rq:at least $\ \deg_G(v)/2 - 1$ nbrs have been trimmed}
	$GOOD(v) = \curly{u \mid u \in \mathcal{N}(v) \ \text{ and }\ \clustid(v) = \clustid(u) }$\;
	
	\lIf{$\clustid(v) \leq n$ and $|GOOD(v)| \leq  \deg_G(v)/2 + 1$}{
		$\texttt{regStatus}(v) \leftarrow \true$
	}
	\caption{$\texttt{trim\_shave}(v)$}
	\label{algo:trimming}
\end{algorithm}

\begin{claim}
	If total number of nodes that are trimmed is bounded by $O(k)$ then \shortOnly{Algorithm }\ref{algo:trimming} runs in $O(k + D\log k)$ rounds.
	\label{lemma:distributed_running_time_trimming_shaving}
\end{claim}
\longOnly{
\begin{proof}
	At the start of \shortOnly{Algorithm }\ref{algo:trimming}, each node is assigned to a vertex group. {\em Trimming} is a iterative process. In each iteration, a node $v$ trims itself if it does not have at least $2/5$ of its neighbours in its group. To decide if a node satisfies the property of $\emph{trimming}$, each node just needs to locally check the neighbour's $\clustid$ which can be done in $O(1)$ rounds. Now if a node $v$ satisfies the criteria of trimming (when it does not have $2/5$ of its neighbours in the same group) then it trims itself from the group and assigns itself a new \clustid, different from other nodes namely $n + \id(v)$.  The node then communicates its trimmed status to all its neighbours. Note that only the neighbors of a node and not all the nodes in a cluster $C$ need to know if a node $v\in C$ is trimmed. All this can be done in $O(1)$ rounds. Further, it is given that at most $O(k)$ nodes could be trimmed. The only difficult part is how to determined when the trimming process has stopped. To do so we use the bound on the number of nodes which can be trimmed as follows: Suppose for some constant $c$ the total number of trimmed nodes is less than $ck$. We assign a leader node $v$ which will track if it is safe to terminate the trimming process. We start by a limit of $l = 2$ rounds. At the end of $l$ rounds, by a simple broadcast, the leader node can find if a there was a node trimmed in last round. This takes $O(D)$ rounds. If there was node which was trimmed, then it allows for doubling of the limit such that $ l = 2l$. This keeps on happening, until it finds that no node was trimmed at the end of the previous process. This coordination takes an extra $O(D \log k)$ overhead. {\em Shaving} is a trivial process which requires each node to check the cluster {\id} of its neighbors. Thus this can be done locally in $O(1)$ rounds. 
\end{proof}}
\shortOnly{Using \cref{{lemma:distributed_running_time_trimming_shaving}}, we get the distributed time complexity of  \shortOnly{Algorithm }\ref{algo:contraction}.}
\begin{lemma}
	The \shortOnly{Algorithm }\ref{algo:contraction} runs in total of $\tilde{O}(n^{1-\epsilon/44} + D)$ rounds in the \CONGEST model to find $MSGC(G,\epsilon)$.
\end{lemma}\longOnly{
\begin{proof}
	Firstly from \cref{corr:sparse_certificate_exists} and \cref{thm:fast_exp_decomp}, we know that the sparse connectivity certificate and the tripartition can be found in sublinear rounds. By the choice of the parameters in \shortOnly{Algorithm }\ref{algo:connect-cert}, the sparse connectivity certificate algorithm takes $\tilde{O}(n^{1-\epsilon/44} + D)$ rounds and $\tripartition$ procedure takes $\tilde{O}(n^{1-\gamma + 10\rho}) = \tilde{O}(n^{1-\epsilon/11})$. In procedure $\tripartition$, we partition the edge set into $\EhExp,E_r$ and $E_s$. Let $\mathcal{X}$ be the set of connected components of the subgraph $G[\EhExp]$. We assign a unique $\clustid$ to every $X \in \mathcal{X}$ known to every vertex in $X$. This is done in a distributed fashion and takes $\max_{X\in \mathcal{X}}Diam(G[X])$ rounds. 
	By \cref{defn:decomposition}, we know that for each $X$, $\Phi(X) \geq  \frac{c}{n^{\rho}}$. Also, we know that the diameter of any graph with expansion $\Phi$ is ${O}(\log n\frac{1}{\Phi})$. Hence the diameter of any component $X$ is ${O}\paren{n^{\rho}\log n}$.	
	As per \cref{lemma:number_trimmed_nodes} the number of trimmed nodes are $O(n^{1-\epsilon/22})$. Hence, by \cref{lemma:distributed_running_time_trimming_shaving}, the distributed algorithm for trimming and shaving takes $O(n^{1-\epsilon/22}  + D\log n)$ rounds. Thus the overall running time is $\tilde{O}(n^{1-\epsilon/44} + D)$ rounds.
\end{proof}}

%% file: camera_ready_sub_parts/min-cut-contracted-graph.tex
%!TEX root = ../main_Min-Cut_latex.tex
\section{Min-Cut in Contracted Graph\label{sec:min-cut-contracted-graph}}
In this section, we show that given a contracted graph $\MSGC(G,\epsilon)$ from \cref{thm:main_contracted_graph_can_be_found}, we can find a min-cut in $O({n}^{1-\eta})$ rounds where $\eta = \Theta(\epsilon)$. Here we use the idea from \cite{karger2000minimum}, which gives a near linear time randomized min-cut algorithm for general graph in the centralized setting. Essentially, \cite{karger2000minimum} illustrates that given a graph, we can construct a set of few spanning trees such that at least one of them crosses a min-cut twice. Further, in each tree \cite{karger2000minimum} can find the cut of minimum size which crosses the tree twice.

The main contribution from this section is two folds. First, in \cref{sub_sec:min-cut-contracted-graph:warmup1}, we show that \cite{karger2000minimum} can be implemented in distributed setting in $\tilde O(n)$ rounds. This is the first ever algorithm which finds exact min-cut in linear time in \CONGEST model. Here, we develop the required machinery which gives a distributed algorithm to do the same in the contracted graph from \cref{thm:main_contracted_graph_can_be_found}, hence giving a sub-linear running time of the algorithm. 
%\mohit{In last feedback danupon asked: why do we have to do this first? does that mean why Sec. 6.1 first? If yes, then does the last line resolve this already? Or we need to say more?}

%% file: camera_ready_sub_parts/warmUp1.tex
\subsection{Min-Cut in General Graph\label{sub_sec:min-cut-contracted-graph:warmup1}}
In this section, we give an algorithm for finding min-cut in weighted graphs. A widely used assumption in the $\CONGEST$ model is that, each edge weight is in $\curly{1,2,\ldots,\poly(n)}$. This allows to exchange edge weight between any two nodes in a single round of communication. Further, this limits the size of any cut to $\poly(n)$, hence can be represented in $O(\log n)$ bits. 
\begin{thm}
		Given a weighted simple graph $G=(V,E)$, with weight function \sloppy $w:E \rightarrow \curly{1,2,\ldots,\poly(n)}$, in $\tilde{O}(n)$ w.h.p. (i) every node knows the network's edge connectivity $\lambda$, and (ii) there is a cut $C$ of size $\lambda$ such that every node knows which of its incident edges are in $C$.
	\label{thm:min_cut_o_n_general_graph}
\end{thm}

 We replace an edge $e$ with weight $w(e)$ with $w(e)$ parallel edges. But the total communication across all these edges in any given round, is still restricted to $O(\log n)$ bits. Let $T$ be spanning tree of $G$. We say that a cut in $G$, $k$-\emph{respects} a spanning tree $T$ if it cuts at most $k$ edges of the tree. \mohit{changed the defn as per karger. so that this is more standardized, so from now cut will respect a tree and not nthe other way round.}
\begin{lem}	
	Given a graph $G$, in $\tilde O(\sqrt{n} + D)$ rounds, we can find a set of spanning trees $\tree = \curly{T_1,\ldots,T_{k}}$ for some $k = {\Theta}(\log^{2.2} n)$ such that w.h.p. there exists a min-cut of $G$ which 2-respects at least one spanning tree $T \in \tree$. Also each node $v$ knows which edges incident to it are part of the spanning tree $T_i$ for $1\leq i \leq k$.
	\label{lem:set_of_spanning_trees}
\end{lem}
The proof of \cref{lem:set_of_spanning_trees} is based on \emph{tree packing}, where a set of $\Theta(\log^{2.2}n)$ MSTs are constructed by appropriately assigning weights to the edges. This is based on \cite{karger1999random,Thorup07} and details of which \longOnly{are left to \cref{appendix:proof_section_6}}\shortOnly{are given in full version of this paper \cite{distributed-min-cut-full}}. The important step of our algorithm is a sub-routine, which given any spanning tree $T$, finds a minimum-sized cut which 2-respects the tree $T$. 
We run this sub-routine on all the trees in the set of trees $\tree$ received from \cref{lem:set_of_spanning_trees}. 
For all the spanning trees $T$ in $\tree$ we  fix an arbitrary root denoted by $r_T$.
For any vertex $v$ other than the root, we use $\parent[T]{v}$ to denote the parent of vertex $v$ and $\edge[T]{v}$ to denote the tree edge $(\parent[T]{v}, v)$. 
We use $\desc[T]{v}$ to denote the set of decedents of the vertex $v$ in tree $T$ and let $\ancestor[T]{v}$ be the set of ancestors of a node $v$ in the spanning tree $T$ including $v$ itself and let $\children_T(v)$ be the set of child nodes of $v$ in the rooted spanning tree $T$. Also let $Depth(T)$ be the distance from root $r_T$ to the furthest node. We give the following lemma which describes high level distributed algorithms in $\CONGEST$ model. These are standard algorithms and details \longOnly{are left to \cref{appendix:proof_section_6}.}\shortOnly{are given in full version of this paper \cite{distributed-min-cut-full}.}
\begin{lemma}
	Let $T$ be a rooted spanning tree of $G$ then,
	\begin{enumerate}
		\item If each node $v$ of $G$ has a message $\msg_{v}$ to be sent to each and every node in $\desc[T]{v}$, then to deliver all such messages it takes $O(Depth(T))$ rounds.   
		\label{point:downcast_messages_general}
		
		\item Let $f:V \rightarrow \curly{0,1,\ldots,\poly(n)}$ and \sloppy$g:V^2\rightarrow \curly{0,1,\ldots,\poly(n)}$ be some functions. Let $f(v) = \sum_{x\in \desc[T]{v}}g(v,x)$ and $g(v,x)$ is precomputed by every node $x$ for all $v \in \anc_{T}(x)$. \monika{what do you mean by "for all $v$ in $anc_T(x)$"} \mohit{$\anc_T\paren{x}$ is a set of all the ancestors of $x$, defined above} Then in total of $O(Depth(T))$ rounds $f$ can be computed by all nodes $v$.
		
		\item  Let $f,g: V \rightarrow \curly{0,1,\ldots,\poly(n)}$ be some functions. For every node $v$ of $G$, let $f(v) = g(v) + \sum_{c \in \child[T]{v}} f(c)$. If $g(v)$ is precomputed by every node $v$, then $f$ can be computed in total of $O(Depth(T))$ rounds by every node $v$ of $G$. Further if we have $k$ such functions  then every node $v$ of $G$ can compute them in $O(Depth(T) + k)$ rounds\mohit{prelim:talk about depth}.
	\end{enumerate}
\label{fact:congest_general}
\end{lemma}
For any vertex set $A \subset V$, let $\partial(A)$ denote the cut induced by $A$. Further, let $C(A) \triangleq |\partial(A)|$ and $C(A,B) \triangleq |\partial(A) \cap \partial(B)|$ for any $A,B \subset V$. We use the operator $\oplus$ to represent the set symmetric difference. In this section, the vertex sets we focus on will be based on some rooted spanning tree $T$. Recall that for any vertex $v$, $\desc[T]{v}$ is the vertex set containing all the decedents of vertex $v$ in the rooted spanning tree $T$. Note that for any two vertices $v,u \in V$, the vertex sets $\desc[T]{v}$ and $\desc[T]{u}$ are either disjoint or one of them is contained in the other.  \longOnly{In most of the proofs given in this section we will have these two cases as illustrated in the \cref{figure:two_cases_section_6}.}
\longOnly{
\begin{figure}[h]
	\centering
	\includegraphics[scale=3.0]{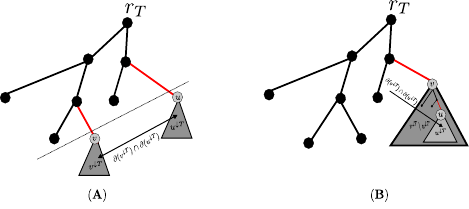}
	\caption{Two different cases illustrating a cut which two respects a spanning tree $T$. In figure $(A)$  the cut is induced by the vertex set $\desc[T]{v} \oplus \desc[T]{u} = \desc[T]{v} \cup \desc[T]{u}$ and in $(B)$ by $\desc[T]{v} \oplus \desc[T]{u} = \desc[T]{v} \setminus \desc[T]{u}$.}
	\label{figure:two_cases_section_6}
\end{figure}\mohit{make (A) and (B) of smaller font}}

\begin{lem}
	For any rooted spanning tree $T$ and for any two nodes $v,u \in V\setminus \curly{r_T}$ we have 
	\begin{enumerate}
		\item $|\partial(\desc[T] v \oplus \desc[T] u)| = C(\desc[T]{v}) + C(\desc[T]{u}) - 2C(\desc[T]{v},\desc[T]{u})$
		\item $E[T] \cap \partial(\desc[T]{v}\oplus \desc[T]{u}) =  \curly{\edge[T]{v},\edge[T]{u}}$ 
	\end{enumerate} 
	\label{lem:imp_characteristics}
\end{lem}\longOnly{
\begin{proof}
	Here we have two cases $\desc[T]{v} \cap \desc[T]{u} = \phi$ or $\desc[T]{v} \cap \desc[T]{u} \neq \phi$. Taking the first case, let's assume that $\desc[T]{v} \cap \desc[T]{u} = \emptyset$.  
	Since, $\oplus$ is the symmetric difference operator, hence, $\desc[T]{v}\oplus \desc[T]{u} = \desc[T]{v} \cup \desc[T]{u}$ as given in \cref{figure:two_cases_section_6}(A). 
	Thus  $\partial(\desc[T]{v}\oplus \desc[T]{u}) = \partial(\desc[T]{v} \cup \desc[T]{u})$. And the edges in $\partial(\desc[T]{v} \cup \desc[T]{u})$ have exactly one end-point in the vertex set $ \desc[T]{v}\cup \desc[T]{u}$. 
	As per definition,  $C(\desc[T]{v}) = |\partial(\desc[T]{v})|$ and are the number of edges which have exactly one end point in the vertex set $\desc[T]{v}$. 
	Also, since $\desc[T]{v} \cap \desc[T]{u} = \phi$, thus $C(\desc[T] v, \desc[T] u)$ is the number of edges which have one end point in ${\desc[T] v}$ and other end point in $\desc[T] u$. 
	Thus the number of edges which have exactly one end point in $\desc[T] v$ and not in $\desc[T] u$ are $C(\desc[T]{v}) - C(\desc[T] v, \desc[T] u)$. 
	Similarly, the number of edges which have exactly one end point in $\desc[T] u$ and not in $\desc[T] v$ are $C(\desc[T]{u}) - C(\desc[T] v, \desc[T] u)$. 
	Adding these two, the total number of edges which have exactly one end point in either $\desc[T]{v}$ or $\desc[T]{u}$, but not both are $C(\desc[T]{v}) + C(\desc[T]{u}) - 2C(\desc[T] v, \desc[T] u)$. Note that there is only one edge $\edge[T]{v}$ connecting $\desc[T]{v}$ to the tree $T$. Also, this is part of the cut. Similarly for $\edge[T]{u}$.

	For second case, when $\desc[T]{v} \cap \desc[T]{u} \neq \phi$. Since, we have an underlying tree $T$ thus either $v \in \desc[T]{u}$ or $u \in \desc[T]{v}$. WLOG, let's assume $u \in \desc[T]{v}$ and hence $\desc[T]{u} \subset \desc[T]{v}$. Thus, we have, $\desc[T]{v} \oplus \desc[T]{u} = \desc[T]{v} \setminus \desc[T]{u}$ as given in \cref{figure:two_cases_section_6}(B).
	Hence, the edges that are in the cut $\partial(\desc[T] v \oplus \desc[T] u)$ have exactly one end point in $\desc[T]{v} \setminus \desc[T]{u}$. 
	Also, $C(\desc[T]{v},\desc[T]{u}) = |\partial(\desc[T]{v}) \cap \partial(\desc[T]{u})|$ are the number edges which have one end point in $\desc[T]{u}$ and other in the vertex set $V \setminus \desc[T]{v}$ as shown by the arrow in \cref{figure:two_cases_section_6}(B).
	Thus $C(\desc[T]{v}) - C(\desc[T]{v},\desc[T]{u})$ is the number of edges with one end point in the vertex set $\desc[T]{v} \setminus \desc[T]{u}$ and other in the vertex set $V \setminus \desc[T]{v}$. 
	But the cut $\partial(\desc[T] v \oplus \desc[T] u)$ also have those edges which have one endpoint in the vertex set $\desc[T]{v}\setminus \desc[T]{u}$ and the other end point in the vertex set $\desc[T]{u}$. And the total number of such edges is $C(\desc[T]{u}) - C(\desc[T]{v},\desc[T]{u})$. Hence  $|\partial(\desc[T] v \oplus \desc[T] u)| = C(\desc[T]{v}) + C(\desc[T]{u}) - 2C(\desc[T]{v},\desc[T]{u})$. And lastly, similar to the previous case, no other tree edge apart from $\curly{\edge[T]{v},\edge[T]{u}}$  has one end point in $V \setminus (\desc[T]{v} \setminus \desc[T]{u})$ and the other in $\desc[T]{v} \setminus \desc[T]{u}$.
\end{proof}}
From \cref{lem:imp_characteristics} it is clear that for any spanning tree $T$ and for any nodes $v,u \in V\setminus T$ the cut $\partial(\desc[T] v \oplus \desc[T] u)$ shares only two tree edges $\curly{\edge[T]{v},\edge[T]{u}}$. Hence if we know $|\partial(\desc[T] v \oplus \desc[T] u)|$ for all $v,u$ then we can find the size of all the cuts which 2-respects the tree $T$ and hence the minimum. In order to find $|\partial(\desc[T] v \oplus \desc[T] u)|$, for any two nodes $v,u$, we will first ensure that every node $v$ finds $C(\desc[T]{v})$. Further, for every node $u \in V\setminus \desc[T]{v}$ we will make sure that node $v$ finds $C(\desc[T]{u})$ and $C(\desc[T]{u},\desc[T]{v})$. Thus for every pair of tree edges $\curly{\edge[T]{v},\edge[T]{u}}$, we have at least one node which can find $|\partial(\desc[T] v \oplus \desc[T] u)|$.  For any rooted spanning tree $T$, let $\child[T]{v}$ be children of node $v$ in $T$. The following simple observation will be handy in finding these information.

\begin{observation}For any rooted spanning tree $T$, we have
	\begin{enumerate}
		\item \sloppy	$\forall v\in V\setminus r_T,\ C(\desc[T]{v}) = \sum_{x \in \desc[T]{v}}C(x,\desc[T]{v})
		 = C(\desc[T]{v},\curly{v}) + \sum_{c \in \child[T]{v}}C(\desc[T]{v},\desc[T]{c})$
		\item  \sloppy$\forall v \in V\setminus r_T$ and $u \in V\setminus \desc[T]{v},\ C(\desc[T]v,\desc[T] u ) = \sum_{x \in \desc[T]{v}}C(x,\desc[T]{u}) = C(\curly{v},\desc[T] u) + \sum_{c \in \child[T]{v}}C(\desc[T]{c},\desc[T] u )$
		.
	\end{enumerate}
	\label{obs:required_quantities_are_sum}
\end{observation}

\begin{claim} \label{lemma:atomic_values_can_be_found}Let $T$ be a rooted spanning tree. Any node $u$ can locally\monika{what do you mean by internally?}\mohit{\sout{internally} and made locally} find $C(\desc[T]{v},u)$ for all $v \in V\setminus \desc{u}$, if the node $u$ knows the set $\ancestor[T]{u}$ and $\ancestor[T]{x}$ for each of its neighbors $x$.
\danupon{Does it mean that $u$ will know all IDs of nodes $v \in V\setminus \desc{u}$?}	
%	\danupon{One thing that is worth noting: $u$ might not know if $v$ exists in the network at all. Nevertheless, $u$ can still tell the value of $C(\desc[T]{v},u)$ (basically it will be zero).}
\end{claim}\longOnly{
\begin{proof}
\begin{algorithm}[h]
	\DontPrintSemicolon
	\lIf{$v \notin \anc_T(u)$}{
		$C(\desc[T]v,u) = \sum\limits_{\substack{(x,u) \in E\\ v\in \ancestor[T]{x}}} w((x,u))$   %|\curly{x\ |\ (x,u) \in E, v\in \ancestor[T]{x}}|$		
	}\lElse{
		$C(\desc[T]v,u) = \sum\limits_{\substack{(x,u) \in E\\ v\notin \ancestor[T]{x}}} w((x,u))$   %|\curly{x\ |\ (x,u) \in E, v\notin \ancestor[T]{x}}|$
	}
	\caption{Finding $C(\desc[T]{v},x)$ at node $x$}
	\label{algo:claim_atomic_values_can_be_found}
\end{algorithm}\monika{This is just a formula such as 6.5? Why do you call it an algorithm?}\mohit{Danupon: Should I make it an observation like \cref{obs:required_quantities_are_sum}?}
We claim that $u$ can compute $C(\desc[T]{v},u)$ using \shortOnly{Algorithm }\ref{algo:claim_atomic_values_can_be_found}. We will now prove its correctness. Recall that $C(\desc[T]{v},u)= |\partial(\desc[T]{v}) \cap \partial(u)|$.
For each node  $v \notin \desc[T]{u}$ that $u$ wants to compute $C(\desc[T]{v},u)$, there are two cases to consider. 

\medskip\noindent {\em Case 1: $v \notin \anc_T(u)$}, i.e.  $v$ is not an ancestor of $u$ (see \cref{figure:two_cases_section_6}(A) for an illustration). Here we have
\begin{align}
\partial(\desc[T]{v}) \cap \partial(u) &= \left\{(a,b) | (a,b) \in E, a \notin \desc[T]{v}, b \in \desc[T]{v}\right\} \cap \Big\{(u,x) | (u,x) \in E \Big\}\nonumber\\
&= \curly{(u,x)\ |\ (u,x) \in E, x \in \desc[T]{v}} & \mbox{(since $\ u\notin \desc[T]{v}$)}\nonumber\\
&= \curly{(u,x)\ |\ (u,x) \in E, v \in \anc_{T}(x)}
\label{eqn:claim_atomic_values_can_be_found_1}
\end{align}

Thus, the first line of \shortOnly{Algorithm }\ref{algo:claim_atomic_values_can_be_found} computes $C(\desc[T]{v},u)$  correctly in this case.\danupon{Where in the inequality you use the fact that $v$ is not a descendant of $u$?}

\medskip\noindent {\em Case 2: $v \in \anc_T(u)$}, i.e.  $v$ is  an ancestor of $u$ (see \cref{figure:two_cases_section_6}(B) for an illustration). We have
\begin{align}
\partial(\desc[T]{v}) \cap \partial(u) &= \left\{(a,b) | (a,b) \in E, a \notin \desc[T]{v}, b \in \desc[T]{v}\right\} \cap \Big\{(u,x) | (u,x) \in E \Big\}\nonumber\\
&= \curly{(u,x)\ |\ (u,x) \in E, x \notin \desc[T]{v}} & \mbox{(since $\ u\in \desc[T]{v}$)}\nonumber\\
&= \curly{(u,x)\ |\ (u,x) \in E, v \notin \anc_{T}(x)}. 
\label{eqn:claim_atomic_values_can_be_found_2}
\end{align}
Thus, the second line of \shortOnly{Algorithm }\ref{algo:claim_atomic_values_can_be_found} computes $C(\desc[T]{v},u)$  correctly in this case. This completes the proof of \cref{lemma:atomic_values_can_be_found}.
\end{proof}}
Note that \cref{lem:imp_characteristics,obs:required_quantities_are_sum,lemma:atomic_values_can_be_found} hold for any spanning tree $T$. Now, we will give  \cref{obs:ancestor_list_can_be_found,lemma:atomic_value_type_1,lemma:atomic_value_type_2,claim:tree_T_delta_u_plus_v} the purpose of which is to prove that every node $v$ can find $C(\desc[T]{u}),C(\desc[T]{v})$ and $C(\desc[T]{v},\desc[T]{u})$ for all $u \in V\setminus \desc[T]{v}$. These claims are an application of \cref{fact:congest_general}. They also use the characterization given in  \cref{obs:required_quantities_are_sum,lemma:atomic_values_can_be_found}.\shortOnly{The detailed proofs of these claims are available in full version of this paper \cite{distributed-min-cut-full}.}
\begin{claim}
	Given a rooted spanning tree $T$, in $O(Depth(T))$ rounds, every node $v$ can find $\ancestor[T]{v}$ and also for all the non-tree neighbors $u$ of $v$, node $v$ can find $\ancestor[T]{u}$.
	\label{obs:ancestor_list_can_be_found}
\end{claim}\longOnly{
\begin{proof}
	Firstly, in $O(Depth(T))$ rounds, each node $v$ can find $\ancestor[T]{v}$ by \cref{fact:congest_general}(1). Now, for any node $u$, $|\ancestor[T]{u}|  \leq Depth(T)$ hence in $O(Depth(T))$ rounds any non-tree neighbor $v$ of $u$ can receive $\ancestor[T]{u}$
\end{proof}}

\begin{claim}
	Let $T$ be a rooted spanning tree. In $O(Depth(T))$ rounds, every node $v \in V$ can find $C(\desc[T]{v})$.
	\label{lemma:atomic_value_type_1}
\end{claim}\longOnly{
\begin{proof}
Let us fix a node $v$. As per \cref{obs:required_quantities_are_sum}, we know that for any node $v$, $C(\desc[T]{v}) =\sum_{x\in\desc[T]{v}} C(\curly{x},\desc[T]{v})$. Based on \cref{lemma:atomic_values_can_be_found} and \cref{obs:ancestor_list_can_be_found}, we know that, in $O(Depth(T))$ rounds, each node $x$ can find $C(x,\desc[T]{v})$ for all the ancestors $v \in \ancestor[T]{x}$. Hence computing $C(\desc[T]{v})$ takes $O(Depth(T))$ rounds by \cref{fact:congest_general}(2). 
\end{proof}}

\begin{claim}
	Let $T$ be a rooted spanning tree. In $O(n)$ rounds, every node $v \in V$ can find $C(\desc[T]{u})$ and $C(\desc[T]{u},\desc[T]{v})$ for all $u \in V \setminus \desc[T]{v}$.	
	\label{lemma:atomic_value_type_2}
\end{claim}\longOnly{
\begin{proof}
	Every node $v$ broadcasts $C(\desc[T]{v})$ computed from  \cref{lemma:atomic_value_type_1}. Since there are $O(n)$  many such messages, this can be done in $O(n + D)$ rounds. 
	
	Let us fix a node $u$, we will show that for all $v$ such that $u \in V \setminus \desc[T]{v}$, node $v$ can find $C(\desc[T]{u},\desc[T]{v})$ in $O(Depth(T))$ rounds. 
	
	By \cref{obs:required_quantities_are_sum}(2), we know that $C(\desc[T]{u},\desc[T]{v}) = \sum_{x \in \desc[T]{v}}C(\desc[T]{u},\curly{x})$. Let choose an arbitrary $x\in \desc[T]{v}$. Since $u \in V\setminus \desc[T]{v}$ thus $u \in  V\setminus \desc[T]{x}$; this allows us to invoke \cref{lemma:atomic_values_can_be_found}  and \cref{obs:ancestor_list_can_be_found}	 to make sure that each $x \in \desc[T]{v}$ can find $C(\desc[T]{u},\curly{x})$ in $O(Depth(T))$ rounds. Thus by \cref{fact:congest_general}(3), we can find $C(\desc[T]{u},\desc[T]{v})$ for all $u \in V \setminus \desc[T]{v}$ in $O(Depth(T))$. There could be at most $n$ such nodes $u$. Thus every node $v \in V$ can find $C(\desc[T]{u},\desc[T]{v})$ for all $u \in V \setminus \desc[T]{v}$ in $O(Depth(T) + n)$ rounds.	
\end{proof}}
\begin{claim}
	Let $T$ be a rooted spanning tree.  In $O(n)$ rounds, for any two nodes $v,u \in V\setminus \curly{r_T}$, one of $v$ or $u$ can find $|\partial(\desc[T]{u} \oplus \desc[T]{v})|$.
	\label{claim:tree_T_delta_u_plus_v}
\end{claim}\longOnly{
\begin{proof}
	Firstly, for any two nodes $u$ and $v$, by \cref{lemma:atomic_value_type_2}, in $O(n)$ rounds both of them know $C(\desc[T]{v})$ and  $C(\desc[T]{u})$. In this proof we will show that at least one of $v$ or $u$ will be able to find $C(\desc[T]{u},\desc[T]{v})$ in $O(n)$ rounds using \cref{lemma:atomic_value_type_2}. Thus the same node can also find $|\partial(\desc[T]{u} \oplus \desc[T]{v})|$ by \cref{lem:imp_characteristics}(1).  
	
	Here again we will consider the two cases illustrated in \cref{figure:two_cases_section_6} that is either $\desc[T]{v} \cap \desc[T]{u} = \emptyset$ or  $\desc[T]{v} \cap \desc[T]{u} \neq \emptyset$. Firstly, lets consider  $\desc[T]{v} \cap \desc[T]{u} = \emptyset$. In this case, $u \in V \setminus \desc[T]{v}$ and $v \in V \setminus \desc[T]{u}$, thus both of them know $C(\desc[T]{u},\desc[T]{v})$ by \cref{lemma:atomic_value_type_2}. Secondly, when $\desc[T]{v} \cap \desc[T]{u} \neq \emptyset$, here WLOG consider that $v \in \desc[T]{u}$. Hence, $u \in V \setminus \desc[T]{v}$. Again from  \cref{lemma:atomic_value_type_2}, node $v$ can find $C(\desc[T]{u},\desc[T]{v})$.  
	\end{proof}}
From \cref{obs:ancestor_list_can_be_found,lemma:atomic_value_type_1,lemma:atomic_value_type_2,claim:tree_T_delta_u_plus_v} we get the following lemma.
\begin{lemma}
	Let $u,v \in V$  be two nodes, let $T$ be any spanning tree $G$, if any node $x \in V$ knows $\anc_{T}(x)$ and $\anc_{T}(y)$ then it can find the edges incident to it which are part of the cut $\partial(\desc[T]{u}\oplus \desc[T]{v})$ in $O(D)$ rounds.
	\label{lemma:edges_can_be_found}
\end{lemma}\longOnly{
\begin{proof}
Some leader node $z$ broadcasts a message to every node to find the edges incident to them which are part of the cut $\partial(\desc[T]{u}\oplus \desc[T]{v})$. Node $u,v$ also receives such a message. On receiving such a message, node $v$ broadcast to all the nodes if $u$ is in $\anc_T(v)$ set. Similarly, $u$ broadcasts if $v$ is in $\anc_T(u)$. We consider two cases illustrated in \cref{figure:two_cases_section_6}. Firstly let $\desc[T]{v} \cap \desc[T]{u} = \emptyset$. Here the cut is given by $(\desc[T]{v} \cup \desc[T]{u}, V \setminus \desc[T]{v} \cup \desc[T]{u})$. Here both $u$ and $v$ broadcasts that the other node is not in its ancestor set. For any two nodes $x,y$ such that $(x,y) \in E$, then the edge $(x,y) \in \partial(\desc[T]{u}\oplus \desc[T]{v})$ iff one of $u$ or $v$ is an ancestor of $x$ and not of $y$ or vice versa. This can be determined by both $x,y$ using $\anc_T(x)$ and $\anc_T(y)$.

In the second case when $\desc[T]{v} \cap \desc[T]{u} \neq \emptyset$, then WLOG let $\desc[T]{u} \subset \desc[T]{v}$. Here node $u$ broadcasts that $v$ is in $\anc_T(u)$. And $v$ broadcasts that $u$ is not in $\anc_T(v)$. The cut, in this case, is given by $(\desc[T]{v} \setminus \desc[T]{u}, V\setminus(\desc[T]{v} \setminus \desc[T]{u}) )$. For any two nodes $x,y$ such that $(x,y) \in E$, then the edge $(x,y) \in \partial(\desc[T]{u}\oplus \desc[T]{v})$ iff exactly one of $x$ or $y$ has $v$ in its ancestor set and not $u$. This also be determined by both $x,y$ using $\anc_T(x)$ and $\anc_T(y)$.
\end{proof}}
\begin{proof}[Proof of \cref{thm:min_cut_o_n_general_graph}]
	Firstly, from \cref{lem:set_of_spanning_trees} we can find a set of spanning trees $\tree$ of size $O(\log^{2.2} n)$ such that at least one of them 2-respects a min-cut. Having received a set of spanning trees, our task is to find the size of minimum cut in each one of them which shares at most two edges with the tree. Let us fix a tree $T$ in the set of spanning trees $\tree$. Our goal is to find the size of the minimum cut which shares two edges with the tree $T$. Firstly, by \cref{lem:imp_characteristics}(2),  we know that for any two nodes $u,v \in V\setminus \curly{r_T}$, $E[T] \cap \partial(\desc[T]{v}\oplus\desc[T]{u}) =  \curly{\edge[T]{v},\edge[T]{u}}$. Hence the value of the minimum cut which 2-respects the tree $T$ is $\min_{\forall u,v} |\partial(\desc[T]{u} \oplus \desc[T]{v})|$. From \cref{claim:tree_T_delta_u_plus_v}, for a fixed rooted spanning tree $T$, in $O(n)$ rounds for any two nodes $u,v$ at least one of them know $|\partial(\desc[T]{u} \oplus \desc[T]{v})|$. Hence, $\min_{\forall u,v} |\partial(\desc[T]{u} \oplus \desc[T]{v})|$ can be found in $O(D)$ rounds. We do this across all the trees. Thus we can find the size of the minimum among all cuts which 2-respects the trees in the set $\tree$ in $O(n \log^{2.2} n)$ rounds. Also, by \cref{lemma:edges_can_be_found}, all the edges incident to any node $x$ in this min-cut can be found locally by node $x$.
\end{proof}

%% file: camera_ready_sub_parts/sec6_final.tex
\subsection{Min-Cut of the Contracted Graph from \cref{thm:main_contracted_graph_can_be_found}}

In this section, our goal is to find the min-cut in the contracted graph $\overline{G} = \MSGC(G,\epsilon)$ given by \cref{thm:main_contracted_graph_can_be_found}. We will follow the same idea as in the previous subsection; that is use \cref{lem:set_of_spanning_trees} to find a set of spanning trees such that a min-cut shares only two edges in one of them. Further we will give a lemma similar to \cref{fact:congest_general} for contracted graph $\overline{G}$ and spanning trees $\overline{T}$. 
 
 \begin{thm}
 	For any  $\epsilon \in (0,\frac{1}{2})$ such that $\delta = n^{2\epsilon}$, let $\overline{G} = \MSGC(G,\epsilon)$ as given by  \cref{thm:main_contracted_graph_can_be_found}. Then w.h.p. $\tilde{O}(D + {n}^{1-\epsilon/22})$ rounds (i) every node knows the edge connectivity $\lambda$ of $\overline{G}$, and (ii) there is a cut $C$ of size $\lambda$ such that every node knows which of its incident edges are in $C$.
 	\label{thm:min-cut-in-contracted-graph}
 \end{thm}
 
Let $\overline{V}$ be the vertex set in this contracted graph $\overline{G}$ and $\overline{E}$ be the remaining edge set after contraction. Below we give a lemma similar to \cref{lem:set_of_spanning_trees} which gives us a set of spanning tree for the graph $\overline G$.
 
 \begin{lem}
 		Given a contracted graph $\overline G$, in total of $\tilde O(\sqrt{n} +D)$ rounds, we can find a set of spanning trees $\overline\tree = \curly{\overline T_1,\ldots,\overline T_{{\Theta}(\log^{2.2} n)}}$ such that w.h.p there exists at least one spanning tree $\overline T \in \overline \tree$, which 2-respects a min-cut of $\overline G$.
 		\label{lem:set_of_spanning_trees_contracted_graph}
 \end{lem}\longOnly{
 \begin{proof}
 	The proof follows from \cref{lem:set_of_spanning_trees} where we constructed $O(\log^{2.2} n)$ MSTs. Here we are only left to show, how an MST can be constructed in the contracted graph. Recall that by \cref{thm:main_contracted_graph_can_be_found}, each vertex $v$ knows if it is in $\core(C)$ of some cluster $C \in \mathcal{C}$ of $\MSGC(G,\epsilon)$. In $O(1)$ rounds, it can find all its neighbours $u$ which are in the same core by finding their $\clustid(u)$. The weights of edges going between vertices of same core are locally set to $\infty$. And then to construct the tree packing, we just require to construct $\Theta(\log^{2.2} n)$ MSTs. This takes $\tilde O(\sqrt{n} +D)$ rounds.
 \end{proof}}
Having received the set of spanning trees, we explain how to find atomic values similar to \cref{lemma:atomic_value_type_2}. One of the key difference here in comparison to  the previous section is the depth of any tree $\overline T \in \overline \tree$. The depth of any such tree w.r.t the contracted graph $\overline G$ is $n^{1-\frac{\epsilon}{22}}$. But for computing the properties (here the size of all induced cuts which share two tree edges) in $O(n^{1-\frac{\epsilon}{22}})$ rounds, we would require the depth to be $O(n^{1-\frac{\epsilon}{22}})$ when the tree is considered w.r.t. the original graph $G$. Here, we map all $T \in \overline{T}$ to a spanning tree of $\overline{G}$, this will enable for efficiently calculating the properties which involve the whole graph. 

A trivial mapping for any spanning tree $\overline T$ of $\overline{G}$ to a spanning tree $T$ of $G$ would be to construct a smallest depth sub-tree in each of the induced subgraph of a core of a cluster. But unfortunately, this does not work because the guarantees we have from \cref{thm:main_contracted_graph_can_be_found} are in terms of the diamater of the clusters and not specifically about the core which might be linear in $n$, even worse, \emph{the subgraph induced by core of a cluster may not even be a connected component}. Thus, here instead of constructing a BFS tree in each subgraph induced by the core of a cluster, we will construct BFS tree in the subgraph induced by the whole cluster. We define this mapping more precisely as below. 

Let $\overline T$ be a rooted spanning tree of $\overline G$. Let $C$ be a non-trivial cluster of $\overline G$ such that $C = S \cup R$,
 where $S= \core(C)$ is the set of vertices corresponding to \emph{core} and  $R = \regnodes(C)$ is the set of regular nodes of cluster $C$. Let $s(C)$ be a vertex of $\overline{G}$ formed by collapsing vertices in $\core(C)$.
 
Now, for every spanning tree $\overline{T}$ of $\overline{G}$, we will define a way to construct a BFS tree in each cluster. 
For every cluster $C$, we assign a leader node $L_{\overline T}(C)$. If for some cluster $C$, $s(C)$ is a root of $\overline T$ then define
  $L_{\overline{T}}(C)$ as any arbitrary node from $\core(C)$. 
Otherwise, if $s(C)$ is any other node of $\overline T$ then define  
   $L_{\overline{T}}(C)$ as a node $r_C \in \core(C)$ such that $(r_C,\parent[\overline T]{s(C)})$ is a tree edge in $\overline{T}$. Note that there will be an unique $r_C$ beecause $s(C)$ has only one parent in the spanning tree $\overline{T}$.
Further, define $\overline{T}[C]$ as a BFS tree in the induced subgraph $G[C]$ rooted at $\parent[\overline T]{r_C}$. 
Now, define the mapping as a multi-set of edges which is the union of these constructed BFS tree edges, preserving the multiplicity: $$\mapping(\overline T,\overline G) \triangleq \overline T \bigcup_{C\text{ is cluster of }\overline G}
      \overline{T}[C]$$
We give the properties of the $\mapping$ in the following lemma.
\begin{lemma}
	Let $\overline T$ be a rooted spanning tree of $\overline G = \MSGC(G,\epsilon)$ received from \cref{thm:main_contracted_graph_can_be_found}. Then  $\operatorname{mapping}(\overline T,\overline G)$ has the following properties
	\begin{enumerate}
		\item If every node $s$ of $\overline G$ (a core node or regular node) has a message $\msg_s$ to be sent to each and every node in $\desc[\overline T]{s}$, then to deliver all such messages it takes $O({n}^{1-\frac{\epsilon}{22}})$ rounds.
		 \label{point:downcast_messages}
		 \item Let $f:\overline{V} \rightarrow \curly{0,1,\ldots,\poly(n)}$ and \sloppy$g:\overline{V}^2\rightarrow \curly{0,1,\ldots,\poly(n)}$ be some functions. Let $f(s) = \sum_{x\in \desc[\overline{T}]{s}}g(s,x)$ and $g(s,x)$ is precomputed by every node $x$ for all $s \in \anc_{\overline T}(x)$. Then in $O(n^{1-\frac{\epsilon}{22}})$ rounds $f$ can be computed by all nodes $s$.
		 
		 \item  Let $f,g: \overline{V} \rightarrow \curly{0,1,\ldots,\poly(n)}$ be some functions. For every node $s$ of $\overline{G}$, let $f(s) = g(s) + \sum_{c \in \child[\overline T]{s}} f(c)$. If $g(s)$ is precomputed by every node $s$, then $f$ can be computed in $O(n^{1-\frac{\epsilon}{22}})$ rounds by every node $s$ of $\overline G$. Further, if we have $k$ such functions  then every node $s$ of $\overline{G}$ can compute  them in $O(n^{1-\frac{\epsilon}{22}} + k)$.

		  \label{point:downcast_}
	\end{enumerate}
\label{lemma:properties_of_mapping}
\end{lemma} 
\longOnly{The proof of \cref{lemma:properties_of_mapping} is similar to \cref{fact:congest_general} and is given in \cref{appendix:proof_section_6}.} Here we use the $\mapping(\overline{T},\overline{G})$ for message passing. The key idea behind this is the fact that there is a path between any two nodes in $G$ of size $O(n^{1-\frac{\epsilon}{22}})$ using the edges of $\mapping(\overline{T},\overline{G})$ and each edge is repeated at most twice in the $\mapping(\overline{T},\overline{G})$ contributed by either one of the BFS tree $\overline{T}[C]$ for some cluster $C$ or by $\overline{T}$ or by both.

Recall that \cref{obs:ancestor_list_can_be_found,lemma:atomic_value_type_1,lemma:atomic_value_type_2,claim:tree_T_delta_u_plus_v} were application of \cref{fact:congest_general} coupled with \cref{obs:required_quantities_are_sum,lemma:atomic_values_can_be_found}. Since, \cref{obs:required_quantities_are_sum,lemma:atomic_values_can_be_found} depend only on hierarchy of nodes established by a spanning tree thus these will be applicable for $\overline{T}$ and $\overline{G}$ as well. Similar to  \cref{fact:congest_general}, for $\overline{G}$ and $\overline{T}$, we have \cref{lemma:properties_of_mapping}.
Thus, \cref{obs:ancestor_list_can_be_found,lemma:atomic_value_type_1,lemma:atomic_value_type_2,claim:tree_T_delta_u_plus_v} can be proved for $\overline{T}$ and $\overline{G}$ as well. This will imply that for any spanning tree $\overline{T}$ and any two vertices $r$ and $s$ of $ \overline{G}$ at least one of them can find $C(\desc[\overline T]{r}),C(\desc[\overline T]{s})$ and $C(\desc[\overline T]{r},\desc[\overline T]{s})$. Hence, using \cref{lem:set_of_spanning_trees_contracted_graph} we can prove \cref{thm:min-cut-in-contracted-graph} similar to  \cref{thm:min_cut_o_n_general_graph}.

%% file: camera_ready_sub_parts/concl.tex
%!TEX root = ../main_Min-Cut_latex.tex

\section{Putting Everything Together \label{sec:concl}}
Here, we prove \cref{thm:main}. Let $\delta = n^{2\epsilon}$ for some $\epsilon \in (0,\frac{1}{2})$. We use a combination of \cref{thm:main_contracted_graph_can_be_found,thm:min-cut-in-contracted-graph} which allows us to find the min-cut of $G$ w.h.p., as required by \cref{thm:main} in $\tilde{O}(D + n^{1-\epsilon/44})$ rounds. Call this algorithm $\mathcal{A}$. Also, by \cref{thm:NS14_main}, we know that min-cut can be found in $\tilde{O}((\sqrt{n}+D)\lambda^{4})$ rounds. We use a combination of both these algorithms. Firstly, it is a well know fact that the approximate diameter $D'$ can be estimated in $O(D)$ rounds in \CONGEST model such that $D \leq D' \leq 2D$.  If $D$ is linear in $n$, then we cannot do much and to find the min-cut we require $\tilde{O}(n)$ rounds for instance by using \cref{thm:min_cut_o_n_general_graph}. Suppose that for some $\mu$, $D' \leq n^{1-\mu}$.

Using the above mentioned parameter, the runtime of \cref{thm:NS14_main} is $\tilde{O}((\sqrt{n}+n^{1-\mu})n^{8\epsilon}) = \tilde{O}(n^{\frac{1}{2}+8\epsilon} + n^{1-\mu + 8\epsilon})$ and the runtime of Algorithm $\mathcal{A}$ is $\tilde{O}(n^{1-\mu} + n^{1-\epsilon/44})$. Firstly, note that both $\mu$ and $\epsilon$ can be determined using a distributed algorithm in $O(D)$ rounds. The runtime of \cref{thm:NS14_main} has two components $n^{\frac{1}{2}+8\epsilon}$ and $n^{1-\mu + 8\epsilon}$ such that when $\mu > 1/2$ the former dominates and when $\mu < 1/2$ then the later.  When $\mu > \frac{1}{2}$, then from runtime complexity of both the algorithms, the first term dominates and the break-point on deciding which among the two algorithms occurs at $\epsilon =\frac{22}{353}$, which leads to $n^{1-\frac{1}{706}}$ contribution from this part. When $\mu \leq \frac{1}{2}$, in this case, in both Algorithm $\mathcal{A}$ and [NS14] the first term dominates. The break-point on deciding which among the two algorithms should occurs at $\epsilon = \frac{44\mu}{353}$. Thus this part gives contributes $n^{1-\frac{1}{353}}D^{\frac{1}{353}}$ to the running time. Combining these two, the runtime complexity of  our algorithm is  $O(n^{1-\frac{1}{706}} + n^{1-\frac{1}{353}}D^{\frac{1}{353}})$.
 

%% file: camera_ready_sub_parts/open.tex
\section{Open Problems}\label{sec:open}

An obvious open problem from our work is whether there are sublinear time distributed algorithms for computing the minimum cut for {\em multi-graphs}, where parallel edges allow more communication per round, and ultimately for {\em weighted graph}, where edge weights do not affect communication. Recall that we showed an $\tilde O(n)$ bound for these problems in \Cref{sub_sec:min-cut-contracted-graph:warmup1}. Note that the same questions are open for centralized deterministic algorithms, where we borrow some techniques from \cite{kawarabayashi2015deterministic}. Understanding these questions in one setting might shed some light on the other. 

To answer the above, it might help to understand the two-party communication complexity of the following minimum cut problem: Nodes of a graph $G=(V, E)$ are partition into two sets, denoted by $V_A$ and $V_B$. Let $C=E(V_A, V_B)$. There are two players, Alice and Bob, who know the information about all edges incident to $V_A$ and $V_B$, respectively. Can Alice and Bob compute the value of the minimum cut of $G$ by communicating $\tilde O(n^{1-\epsilon}|C|)$ bits? A negative answer to this question would imply a lower bound in the CONGEST model by a standard technique (e.g. \cite{FrischknechtHW12,AbboudCK16,Censor-HillelKP17,Nanongkai-STOC14}). A positive answer would rule out pretty much the only known technique to prove lower bounds and might lead to a fast algorithm in the CONGEST model, as happened for all-pairs shortest paths \cite{Censor-HillelKP17,BernsteinN19}.

It is also very interesting to show tight bounds for computing the minimum cut on unweighted simple graphs. Since we already achieve sublinear time, past experiences from approximation distributed algorithms suggest that this might be $\tilde \Theta(\sqrt{n}+D)$. An $\tilde O(\sqrt{n}\poly(D))$-time algorithm would be a big step towards this bound. Ruling out such algorithm should be very interesting, since it should imply a bound between $\tilde O(\sqrt{n}+D)$ and $\tilde O(n)$.

A special case that deserves attention is when the graph connectivity is small. For example, is there an algorithm that can check whether an unweighted network has connectivity at most $k$ in $\poly(k,D,\log(n))$ time? A less ambitious goal that is already interesting is to get a  $f(k)\poly(D,\log(n))$-time algorithm, for some function $f$ that is independent of $D$ and $n$ (an algorithm ``parameterized by $k$''). 
%, a bound in the same style as the those from parameterized algorithms. 
%
% Can we avoid the polynomial dependency on $n$? It will be very interesting to, e.g., design a distributed algorithm that can check whether an unweighted network has connectivity at most $k$ in $\poly(k,D,\log(n))$ time. 
%
Bounds in these forms are currently known only for $k\leq 2$ \cite{PritchardT11}.\footnote{Update: We recently learned that such bound can be essentially extended to $k=O(1)$ in the sense that there is a $\poly(D)$-time algorithm \cite{Parter19-smallcut}. We thank Merav Parter for this information.}

As noted earlier, this paper is part of an effort to understand {\em exact} distributed graph algorithms. So far, not many problems admit tight bounds when it comes to exact solutions.  (Minimum spanning tree \cite{GarayKP98,KuttenP98} and all-pairs shortest paths \cite{BernsteinN19} are among a few that we are aware of.) Many problems are yet to be explored, e.g. single-source shortest paths \cite{ForsterN18}, maximum weight/cardinality matching \cite{AhmadiKO18-matching}, st-cut/flow \cite{GhaffariKKLP15}, vertex connectivity \cite{Censor-HillelGK14-decomposition}, densest subgraph \cite{DasSarmaLNT12}, and betweenness and closeness centralities \cite{HoangPDGYPR19}.

A more general question that was raised recently \cite{Censor-HillelKP17} is to classify complexities of global problems in the CONGEST model. Tight bounds witnessed so far are in the form of either $\tilde \Theta(D)$, $\tilde \Theta(\sqrt{n}+D)$, $\tilde \Theta(n)$, or $\tilde \Theta(n^2)$. Are there (preferably natural) graph problems with complexity in-between (e.g. $\tilde \Theta(n^{1/2+\epsilon}+D)$ or $\tilde \Theta(n^{1+\epsilon})$ for some constant $\epsilon>0$)? A bound in the form $\tilde \Theta(n^{1/2}D^\epsilon+D)$ will be also interesting, and we suspect that it might be achievable when the two-party communication rounds are considered (as in \cite{ElkinKNP14}). 

%There is a possibility to prove a lower bound in the form $\tilde \Theta(\sqrt{n^{1/2}}D^{\epsilon}+D)$ by involving rounds of communication (see \cite{ElkinKNP14}), but  

% Deterministic

% - Lower bound of n^2 and better 

%% file: camera_ready_sub_parts/ack.tex
This project has received funding from the European Research Council (ERC) under the European Union's Horizon 2020 research and innovation programme under grant agreement No 715672. Daga, Nanongkai, and Saranurak  were also supported by the Swedish Research Council (Reg. No. 2015-04659). 
The research leading to these 
results has received funding from the European Research Council under 
the European Union's Seventh Framework Programme (FP/2007-2013) / ERC 
Grant Agreement no. 340506.   

%% file: camera_ready_sub_parts/appendix_section3.tex
%!TEX root = ../main_Min-Cut_latex.tex
\section{Proof of \cref{thm:random_skeletons_all_cuts}\label{appendix:proof_random_skeletons_all_cuts}}
Firstly, we state a known fact about the number of cuts of a particular size
\begin{lem}[\cite{karger2000minimum}\ {Theorem 3.2}]
Let $G$ be any unweighted and undirected graph. Let $\lambda$ be the size of the minimum cut. Then for any constant $\alpha \geq 1$, the number of cuts in $G$ which are of the size at most $\alpha \lambda$ is ${n^{\ceil{2\alpha}}}$
\label{lem:number_of_cuts}
\end{lem}

\begin{proof}[Proof of \cref{thm:random_skeletons_all_cuts}]
	Let $\alpha\geq 1$ and let us fix an arbitrary cut $C$ of size at most $\alpha k \leq \alpha \lambda$. Using \cref{lem:number_of_cuts} we know that the number of such cuts is $n^{\ceil{2\alpha}}$ Then using Chernoff bound we have
	$$\Pr[\text{\# edges sampled from }C \geq (1+\epsilon)p\alpha k] \leq e^{-\frac{\epsilon^2p\alpha k}{3}} = e^{-\tau \alpha \ln n}$$
	For the last equality in the above equation, we choose $p = 3\tau \frac{\ln n}{\epsilon^2 k} = \theta(\frac{\ln n}{\epsilon^2 k})$. Also from \cref{lem:number_of_cuts}, we know that the number of such cuts is $n^{2\alpha}$. Thus by union bound we have
	\begin{align}
		&\Pr[\text{\# edges sampled from any cut of size at most $\alpha\cdot k$} \geq ((1+\epsilon)p\alpha k)] &\leq n^{2\alpha}\cdot e^{-\tau \alpha \ln n}\nonumber\\
		&=e^{2\alpha \ln n}\cdot e^{-\tau \alpha \ln n}\nonumber\\
		&=e^{-(\tau-2)\alpha\ln n} \nonumber \\
		&=n^{-(\tau-2)\alpha}
		\label{eqn:union_bound_number_of_edges_sampled}
	\end{align}
	Further, the minimum value of $k$ could be $1$ and the size of any cut is at most $n^{2}$. Thus we have at most $n^2$ values of $\alpha$. Hence using the union bound again \cref{eqn:union_bound_number_of_edges_sampled} for all values of alpha we have $$\Pr[\text{\# edges sampled from any cut $C$} \geq ((1+\epsilon)p|C|)] \leq n^{-(\tau-2)-2}$$
	
	Thus by choosing $\tau \geq 3$ this would imply that w.h.p edges sampled from all the cuts $C$ are less than $(1+\epsilon)p|C|$
\end{proof}

%% file: camera_ready_sub_parts/appendix_section6.tex
\section{Omitted proofs from \cref{sec:min-cut-contracted-graph}\label{appendix:proof_section_6}}
\subsection{Proof of \cref{lem:set_of_spanning_trees_contracted_graph}}
Firstly, in this section we prove \cref{lem:set_of_spanning_trees}. To prove this we review the \emph{greedy tree packing} as given in \cite{Thorup07} and mentioned earlier in \cite{plotkin1995fast,young1995randomized,thorup2000dynamic}. 
\begin{defn}
	For any set of spanning tree $\tree$, let the load of an edge $e$ be defined as $L^{\tree}(e) = |\curly{T \mid e \in T}|$. 
	A set of spanning tree $\tree = \curly{T_1,T_2,\ldots,T_k}$ is a greedy tree packing if each $T_i$ is a 
	minimum spanning tree with respect to the load on each edge given by $L^{\tree_{i-1}}(e)$ where 
	$\tree_{i-1} = \curly{T_1,T_2,\ldots,T_{i-1}}$.
\end{defn}
We now state known results about tree packing
\begin{lemma}[\cite{Thorup07}Lemma~6]
Let $C$ be any cut with $<1.1\lambda$ edges and let $\tree$ be a greedy
tree packing with $\omega(\lambda \ln m)$ trees. Then a fraction $1/3$ of the trees in $\tree$ cross
$C$ at most twice.
\label{lemma:fraction_cuts_crosses_greedy_packing}
\end{lemma}
In the above lemma we are required to construct $\omega(\lambda \ln m)$, which could be linear in $n$. This is too large for our purpose. Thus we use the sampling idea from \cite{karger1999random} which will reduce the size of min-cut to $\omega(\ln m)$
\begin{lemma}
Let $p$ be a probability and $H =G_p$ be a random subgraph
of $G$ including each edge independently with probability $p$. Let $\lambda_H$ be the
edge connectivity of $H$. Suppose $p\lambda = \omega(\log n)$. Then, w.h.p., $\lambda_H = (1 \pm o(1))p\lambda$. Moreover, w.h.p., min-cuts of $G$ are near-minimal in H and vice
versa. More precisely, a min-cut $C$ of $G$ has $(1+o(1))\lambda_H$ cross edges in $H$.
Conversely, a min-cut $C_H$ of $H$ has $(1+o(1))\lambda$ cross edges in $G$.
	
\label{lemma:karger_skeleton}
\end{lemma}
Using the above we can prove \cref{lem:set_of_spanning_trees}. This is similar to  proof of \cite[Lemma 8]{Thorup07}.
\begin{proof}[Proof of \cref{lem:set_of_spanning_trees}]
We need to prove that w.h.p., we can construct a set of spanning trees such that at least one of them 2-respects a min-cut. At the beginning, we do not know $\lambda$, but we know that for some $i$, $\lambda/2^{i} = \Theta(\log^{1.1} n)$. We choose this value of $i$. Let $p = \frac{1}{2^{i}}$. Let $H = G_p$ as given by \cref{lemma:karger_skeleton}. By the same lemma, we know that the edge connectivity $\lambda_H$ of $H$ is $\Theta(\log^{1.1} n)$. Thus $H$ has at least  $\Theta(\log^{1.1} n)$ spanning forests. We also have that any min-cut $C$ of G has $(1+o(1))\lambda_H$ edges. Thus by  \cref{lemma:fraction_cuts_crosses_greedy_packing}, a tree packing $\tree$ of $H$ with $O(\log^{2.2} n)$ trees has a tree $T$ which crosses the min-cut at most twice. To construct a MST in \CONGEST model we require $\tilde{O}(\sqrt{n} + D)$ rounds. Thus this lemma follows.
\end{proof}
\subsection{Proof of \cref{fact:congest_general} and \cref{lemma:properties_of_mapping}} 
\cref{fact:congest_general} gives known algorithms in \CONGEST model. \cref{fact:congest_general}(1) is a simple downcast with pipelined messages.
\begin{proof}[Proof of \cref{fact:congest_general}(1)]
	Here each node $v$ has at most $Depth(T)$ ancestors. Thus it receives at most $Depth(T)$ messages. The idea here is to perform a message passing from top to bottom in a pipelined fashion. At round $t = 0$, the root $r_T$ sends its messages to all its children which immediately send to their children. Subsequently, any internal node $v$ of the tree which receives some $\texttt{msg}$ from $\parent[T]{v}$ in round $t$ immediately sends $\msg$ to all its children in round $t+1$. At round $t=1$ nodes at level $1$ (distance $1$ from root $r_T$) release there messages. At any round $t = t' \leq Depth(T)$, nodes at level $t'$ release there message. Thus in $O(Depth(T))$ rounds all nodes $v$ receive messages $\msg_a$ from all $a \in \anc_T(v)$.
\end{proof}
Let $\level[T]{v}$ of any node $v$ be the distance from the root $r_T$ following the tree edges. For \cref{fact:congest_general}(2),  we give a distributed-two phased procedure in \cref{algo:trsf_phase-2}.  
\begin{algorithm}[h]	
	\DontPrintSemicolon
	\SetKwFor{Forp}{for}{parallely}{endfor}
	\SetKwFor{Forw}{for}{wait}{endfor}
	\SetKwProg{phase}{Phase}{}{}
	\SetKwProg{preprocessing}{Pre-Processing}{}{}
	\SetKwProg{availableInfo}{Available Info:}{}{}
	\preprocessing{}{For any node $x, \forall v \in \ancestor[T]{x}$, node $x$ knows $g(v,x)$ through a pre-processing step} 
	
	\setcounter{AlgoLine}{0}
	\phase{1 : Aggregation phase run on all node $x$, aggregates $g(v,\desc[T]{x}) = \sum_{x'\in \desc[T]{x}}g(v,x')\ \forall v\in \ancestor{x}$}{
		
		\lForw{rounds $t = 1$ to $Depth(T)-\level[T]{x}$}{}
		$l \gets 0$\;
		\For{rounds $t = Depth(T) - \level[T]{x} + 1$ to $Depth\paren{T}$}{
			$v \leftarrow$ ancestor of node $x$ at level $l$\;
			\lIf{$x$ is leaf node}{$g(v,\desc[T]{x}) \gets g(v,x)$}
			\Else{
				\lForp{$c \in \child[T]{x}$}{collect $\angularbrac{l,g(v,\desc[T]{c})}$}
				$g(v,\desc[T]{x})\gets g(v,{x}) + \sum_{c \in \child[T]{x}} g(v,\desc[T]{c})$\;
			} \label{interstep:cut-1-respects-tree-internal-node}
			send to the parent  $\angularbrac{l,g(v,\desc[T]{x})}$\;
			$l \gets l+1$\;
		}	
	} \setcounter{AlgoLine}{0}
	\phase{2: Computation Phase (run on all node $v \in V$), finds f(v)}{
		\lavailableInfo{}{Each node $v$ knows $g(v,\desc[T]{c})$ for all $c \in \child[T]{v}$ }
		\lIf{$v$ is a leaf node}{$f(v) \gets g(v,v)$}
		\Else{
			$f(v) \gets g(v,v) + \sum_{c \in \child[T]{v}} g(v,\desc[T]{c})$
		}
	}

	\caption{Computes $f$, if $f(v) = \sum_{x\in \desc[T]{v}}g(v,x)$}
	\label{algo:trsf_phase-2}
\end{algorithm}
\begin{proof}[Proof of \cref{fact:congest_general}(2)]
	
	For any node $v$, $f(v)$ depends on the value $g(v,x)$ for all $x\in \desc[T]{v}$, thus each such node $x$ \emph{convergecasts} (see \cite[Chapter 3]{peleg2000distributed}) the required information up the tree which is supported by aggregation of the values.  We will give an algorithmic proof for this lemma. The algorithm to efficiently  compute function $f(\cdot)$ is given in \shortOnly{Algorithm }\cref{algo:trsf_phase-2}. 
	
	The aggregation phase of the algorithm given in Phase {1} runs for at most $Depth(T)$ rounds and facilitates a coordinated aggregation of the required values and convergecasts them in a synchronized fashion. 
	Each node $x$ in Phase {2}, 
	sends $\level[T]{x}$ messages of size $O(\log n)$  to its parent, 
	each message include $g(v,{x^{\downarrow T}})$ where $v \in \ancestor[T]{x}$; which as defined earlier is the contribution of nodes in $x^{\downarrow T}$ to $f(v)$. 
	This message passing takes $O(1)$ time since $1 \leq g(v,{x^{\downarrow T}})  \leq \poly(n)$ is of size $O(\log n)$ bits. 
	For brevity, we assume this takes exactly $1$ round, this enables us to talk about each round more appropriately as follows: Any node $x$ at level $\level[T]{x}$ waits for round $t = 1$ to $Depth(T) - \level[T]{x}$. 
	For any $l \in [0, \level[T]{x}-1]$, in round $t = Depth(T) - \level[T]{a} + l + 1$  node  $x$ sends to its parent $\angularbrac{l, g(v,{x^{\downarrow T}})}$ where $v$ is the ancestor of $x$ at level $l$. 
	When node $x$ is an internal node then, $g(v,\desc[T]{x})$  
	depends on $g(v,x)$ which can be pre-calculated. 
	Also, $g(v,{x^{\downarrow T}})$ depends on $g(v,\desc[T]{c})$
	for all $c \in \child[T]{x}$ which are at level $\level[T]{x}+1$ 
	and have send to $x$ (which is their parent) the
	message $\angularbrac{l, g(v,\desc[T]{c})}$ in the $(Depth(T) - \level[T]{x} + l)^\text{th}$ round. 
	For a leaf node $x$, $g(v,{x^{\downarrow T}}) = g(v,x)$ which again is covered in pre-processing step.
	
	In Phase {2}, node $v$ computes function $f(v)$. As per definition of $f$ each internal node $v$ requires $g(v,{\desc[T]{c}})\ \forall c \in \child[T]{v}$ and $g(v.v)$ is computed in the pre-processing step. And $g(v,{\desc[T]{c}})$ is received by $v$ in the aggregation phase. When node $v$ is a leaf node, $f(v)$ depends only on $g(v,v)$ since $\desc[T]{v} = \curly{v}$.
\end{proof}
For \cref{fact:congest_general}(3), we use similar technique as Proof of \cref{fact:congest_general}(1). But instead of sending a train of messages towards the leaf nodes, we send a train of messages towards the root in a synchronized fashion.
\begin{proof}[Proof of \cref{fact:congest_general}(3)]
	Here we are given that $f(v) = g(v) + \sum_{c \in \child[T]{v}}f(c)$. We know that $1\leq f(v) \leq \poly(n)$. Thus, for any node $v$ to send $f(v)$ from one node to another through a physical link it takes $O(1)$ rounds. For brevity let's assume that this takes exactly $1$ round. For this lemma, if we can show that any node $x$ at level $\level[T]{x} = t$ computes $f(x)$ and sends it to $\parent[T]{x}$ in round $Depth(T) - \level[T]{x}$, then $f(v)$ can be computed by each node $v$ in $O(Depth(T))$ rounds.
	For the base case, at round $t = 0$, leaf nodes $x$ such that $\level[T]{x} = Depth(T)$ send $f(x) = g(x)$ (pre computed by $x$) to $\parent[T]{x}$. Fix a $t \leq Depth(T)$, assume that all node $x$ at level $\level[T]{x} = t$ computes $f(x)$ and sends it to $\parent[T]{x}$ in round $Depth(T) - \level[T]{x}$. Now using this information nodes $x$ at $\level[T]{x} = t+1$ can compute $f(x)$ and send it to $\parent[T]{x}$. If $x$ is a leaf node then it has the precomputed value of $f(x) = g(x)$. Otherwise it uses $f(x) = g(x) + \sum_{c\in\child[T]{x}}f(c)$. Recall that by induction hypothesis all children $c$ of $x$ have sent $f(c)$ to $x$ in round $Depth(T) - \level[T]{c} = Depth(T) - \level[T]{x} - 1$. Thus $f(x)$ can be sent to $\parent[T]{x}$ in round $Depth(T) - \level[T]{x}$. 
	
	Further, if we have $k$ such functions $f_1,\ldots,f_k$, then we can use a train of $k$ messages sent by each node $x$ with the values of $f_1(x),\ldots,f_k(x)$ to $\parent[T]{x}$
\end{proof} 

\paragraph{Proof of \cref{lemma:properties_of_mapping}}
The proof of \cref{lemma:properties_of_mapping} is similar to \cref{fact:congest_general}. In a general physical network $G$, for any spanning tree $T$ in $O(1)$ round a node $v$ can send a message of $O(\log n)$ bits to $\parent[T]{x}$ or to all child nodes $c$ in $\child[T]{v}$. But for a contracted graph $\overline{G} = \MSGC(G,\epsilon)$ and a spanning tree $\overline T$ of $\overline G$, a node $s$ of $\overline{G}$ can not send a message of $\log n$ bits to $\parent[\overline  T]{s}$ or to all child nodes $c$ in $\child[\overline T]{s}$ in $O(1)$ rounds because some of these nodes are a set of nodes in the original network $G$ and, thus, are not immediate neighbors of s in G. But due to condition of low diameter \cref{thm:main_contracted_graph_can_be_found}, we can show that in total for computation of any of the functions described  \cref{lemma:properties_of_mapping} we just pay an over head of $O(n^{1-\frac{\epsilon}{22}})$ rounds. Firstly, we prove some properties of $\mapping(\overline{G},\overline{T})$.

Recall that the nodes of $\overline{G}$ are either physical nodes or formed by collapsing $\core(C)$ of some cluster $C$. Also recall, that in each $\core(C)$, we have chosen $r_C$ which is the root of the BFS tree $\overline{T}[C]$. 
\begin{claim}
	\label{claim:small_path_Size}
Any node $v$ of the contracted graph $\overline G$, has a path of length $O(n^{1-\frac{\epsilon}{22}})$ to all nodes $a \in \anc_{\overline T}(v)$ using the edges of $\operatorname{mapping}({\overline G},{\overline T})$. In case $a$ is a vertex formed by collapsing $\core(C)$ of a cluster $C$ then there is a  path from $v$ to $r_C$ of length $O(n^{1-\frac{\epsilon}{22}})$.
\label{}
\end{claim}
\begin{proof}
Let $\mathcal{C}$ be the cluster set of $\overline{G}$. Recall that in $\overline G$ 
there are $\Theta(n^{1-\frac{\epsilon}{22}})$ 
nodes. Thus, any spanning tree $\overline{T}$ of $\overline{G}$ has depth $O(n^{1-\eta_1})$.  Let $v$ be any node. Lets fix an arbitrary $a \in \anc_{\overline T}(v)$. Following the tree edges of $\overline{T}$ we have a path of length $O(n^{1-\eta_1})$ between $a$ and $v$. But this is in the contracted graph $\overline{G}$ and not in the given physical graph $G$. The difference here is that, some of the nodes on this path are formed by collapsing $\core(C)$ of some cluster $C \in \mathcal{C}$. Thus to traverse through such nodes the path uses the BFS tree $\overline{T}[C]$ which is part of $\mapping(\overline G,\overline T)$. As per \cref{defn:Min-Cut-Preserving-Sub-linear-Graph-Contraction}(4) we know that $\sum_{c \in \mathcal{C}}diam(G[C]) = O(n^{1-\frac{\epsilon}{20}})$ . Hence, $\sum_{c \in \mathcal{C}}Depth(\overline{T}[C])  = O(n^{1-\frac{\epsilon}{22}})$.
\end{proof}
The proof \cref{lemma:properties_of_mapping} uses \cref{claim:small_path_Size}. 
\begin{proof}[Proof of \cref{lemma:properties_of_mapping}]
	 Let $C \in \mathcal{C}$ be some cluster. Recall that $s(C)$ is a node formed by collapsing $\core(C)$. Any node $x \in C \setminus \core(C)$  may have some incident edges which are part of both $\overline T[C]$ and $\overline{T}$ at the same time. In any given round, these edges will be tasked to carry a message of two forms by node $x$: messages sent from children of $s(C)$ in $\overline{T}$ to $s(C)$ or messages sent from $x$ to $\parent[\overline T]{x}$. Thus we include an extra label to the message to indicate which one of the two cases it belongs to so that it can be routed appropriately either using the edges of $\overline{T}$ or by $\overline{T}[C]$. This will increase the complexity by a factor of $2$.  Hence using the same arguments in proof of \cref{fact:congest_general} and \cref{claim:small_path_Size} this lemma follows.
\end{proof}

%% file: camera_ready_sub_parts/appendix_exp_decomposition.tex
%!TEX root = ../main_Min-Cut_latex.tex
\section{Proof of \cref{thm:fast_exp_decomp}\label{sec:proof_exp_decomposition}}
\paragraph{Disclaimer}: \textbf{This section is taken almost as is from [CPZ18] except for a few changes of  parameters to suit our need. It is included only for the sake of verification.}

We first introduce some notation.
Let $\deg_H(v)$ be the degree of $v$ in the subgraph $H$,
or in the graph induced by edge/vertex set $H$.
Let $V(E^\ast)$ be the set of vertices induced by the edge
set $E^\ast \subseteq E$.
The {\em strong diameter} of a subgraph $H$ of $G$ is defined as $\max_{u,v \in H} {\operatorname{dist}}_H(u,v)$ and the {\em weak diameter} of  $H$  is $\max_{u,v \in H} {\operatorname{dist}}_G(u,v)$.

The goal of this is to prove \cref{thm:fast_exp_decomp}. The algorithm for \cref{thm:fast_exp_decomp} is based on repeated application of a black box algorithm $\mathcal{A}^\ast$, which is given a subgraph $G'=(V',E')$
of the original graph $G=(V,E)$, where $V' = V(E')$, $n'=|V'|$, and $m'=|E'|$.
In $\mathcal{A}^\ast$, vertices may halt the algorithm at different times.

\paragraph{Specification of the Black Box.}
The goal of $\mathcal{A}^\ast$ is, given $G'=(V',E')$,
to partition $E'$ into $E'=\EhExp'\cup E_s'\cup E_r'$ satisfying some conditions.
The edge set $\EhExp'$ is partitioned into $\EhExp' = \bigcup_{i=1}^t \mathcal{E}_i$.
We write $\mathcal{V}_i = V(\mathcal{E}_i)$ and
$\mathcal{G}_i = (\mathcal{V}_i, \mathcal{E}_i)$,
and define $S=V'\setminus \left( \bigcup_{i=1}^t \mathcal{V}_i \right)$.

\begin{description}
\item[(C1)] The vertex sets $\mathcal{V}_1, \ldots, \mathcal{V}_t, S$ are disjoint and partition $V'$.
\item[(C2)] The edge set $E_s'$ can be decomposed as $E_s'=\bigcup_{v\in S}E_{s,v}'$, where $E_{s,v}'$
is a subset of edges incident to $v$, viewed as oriented away from $v$.
This orientation is acyclic.
For each vertex $v$ such that $E_{s,v}'\neq \emptyset$, we have $|E_{s,v}'|+ \deg_{\EhExp'}(v) \leq n^{\gamma}$. Each vertex $v$ knows the set $E_{s,v}'$.
\item[(C3)] Consider a subgraph $\mathcal{G}_i=(\mathcal{V}_i,\mathcal{E}_i)$.
Vertices in $\mathcal{V}_i$ halt after the same
number of rounds, say $K$.
Exactly one of the following subcases will be satisfied.
\begin{enumerate}
    \item [{\bf (C3-1)}] All vertices in $\mathcal{V}_i$ have degree $\Omega(n^{\gamma})$ in the subgraph $\mathcal{G}_i$,  each connected component of $\mathcal{G}_i$ has $\tilde{O}(n^{\rho})$ mixing time, and $K = \tilde{O}(n^{10\rho})$. Furthermore, every vertex in $\mathcal{V}_i$ knows that they are in this sub-case.
    \item [{\bf (C3-2)}] $|\mathcal{V}_i|\leq n' -\tilde{\Omega}(n^{\gamma})$, and every vertex in $\mathcal{V}_i$ knows they are in this subcase.
\end{enumerate}
\item[(C4)] Each vertex $v \in S$ halts in $\tilde{O}(n'/n^{\gamma})$ rounds.
\item[(C5)] The inequality
$E_r'\leq \Big(|E'| \log|E'| -\sum_{i=1}^t |\mathcal{E}_i| \log |\mathcal{E}_i|\Big)/
(6n^{\rho}\log m)$ is met.
\item[(C6)] Each cluster $\mathcal{V}_i$ has a distinct identifier. When a vertex $v \in \mathcal{V}_i$ terminates, $v$ knows the identifier of  $\mathcal{V}_i$. If $v \in S$, $v$ knows that it belongs to $S$.
\end{description}

We briefly explain the intuition behind these conditions.
The algorithm $\mathcal{A}^\ast$ will be applied recursively to all
subgraphs $\mathcal{G}_i$ that have yet to satisfy the minimum degree
and mixing time requirements specified in \cref{thm:fast_exp_decomp} and \cref{defn:decomposition}.
Because vertices in different components halt at various times,
they also may begin these recursive calls at different times.

The goal of (C2) is to make sure that once a vertex $v$ has $E_{s,v}' \neq \emptyset$, the total number of edges added to $E_{s,v}$ cannot exceed $n^\gamma$.
The goal of (C3) is to guarantee that the component size drops at a fast rate.
The idea of (C5) is that the size of $E_r'$ can be mostly charged to the number of the edges in the small-sized edge sets $\mathcal{E}_i$;
this is used to bound the size of $E_r$ of our graph partitioning algorithm.

Note that in general the strong diameter of a subgraph $\mathcal{G}_i$ can be much higher than the maximum running time of vertices in  $\mathcal{G}_i$, and it could be possible that $\mathcal{G}_i$ is not even a connected subgraph of $G$. However, (C6) guarantees that each vertex $v \in \mathcal{V}_i$ still knows that it belongs to $\mathcal{V}_i$. This property allows us to recursively execute  $\mathcal{A}^\ast$ on each subgraph $\mathcal{G}_i$.

\begin{lemma}\label{le:auxiliary}
There is an algorithm $\mathcal{A}^\ast$ that finds a partition $E'=\EhExp'\cup E_s'\cup E_r'$ meeting the above specification in the $\CONGEST$ model, w.h.p.
\end{lemma}

Assuming \cref{le:auxiliary},
we are now in a position to prove \cref{thm:fast_exp_decomp}.

\begin{proof}[Proof of \cref{thm:fast_exp_decomp}]
Let $\mathcal{A}^\ast$ be the algorithm for \cref{le:auxiliary}.
Initially, we apply $\mathcal{A}^\ast$ with $G'=G$, and this returns a partition $E'=\EhExp'\cup E_s'\cup E_r'$.

For each  subgraph $\mathcal{G}_i$ in the partition output by an invocation of $\mathcal{A}^\ast$, do the following.
	If $\mathcal{G}_i$ satisfies (C3-1), by definition it must have $\tilde{O}(n^{\rho})$ mixing time, and all vertices in $\mathcal{G}_i$ have degree $\Omega(n^{\gamma})$ in $\mathcal{G}_i$; we add the edge set $\mathcal{E}_i$ to the set $\EhExp$ and all vertices in $\mathcal{V}_i$ halt.	
	Otherwise we apply the algorithm recursively to $\mathcal{G}_i$, i.e., we begin by applying $\mathcal{A}^\ast$ to $G' = \mathcal{G}_i$ to further partition its edges. All recursive calls proceed in parallel, but may begin and end at different times. 
	Conditions (C1) and (C6) guarantee that this is possible. (Note that if $\mathcal{G}_i$ is disconnected, then each connected component of $\mathcal{G}_i$ will execute the algorithm in isolation.)
	
Initially $E_r = \emptyset$ and $E_s = \emptyset$. After each invocation of $\mathcal{A}^\ast$, we update $E_r \gets E_r\cup E_r'$, $E_s \gets E_s\cup E_s'$, and $E_{s,u} \gets E_{s,u}\cup E_{s,u}'$ for each vertex $u$.

\paragraph{Analysis.} We verify that the three conditions of \cref{defn:decomposition} are satisfied.
First of all, note that each connected component of $\EhExp$ terminated
in (C3-1) must have $\tilde{O}(n^{\rho})$ mixing time, and all vertices in the component have degree $\Omega({n^{\gamma}})$ within the component.
Condition (a) of \cref{defn:decomposition} is met.
Next, observe that Condition (b) of \cref{defn:decomposition}
is met due to (C2). If the output of $\mathcal{A}^\ast$ satisfies that $E_{s,v}' \neq \emptyset$, then $|E_{s,v}|$ together with the number of remaining incident edges (i.e., the ones in $\EhExp'$) is less then ${n^{\gamma}}$. Therefore, $|E_{s,v}|$ cannot exceed ${n^{\gamma}}$, since only the edges in $\EhExp'$ that are incident to $v$ can be added to $E_{s,v}$ in future recursive calls.
Lastly, we argue that (C5) implies that Condition (c) of \cref{defn:decomposition} is satisfied.
Assume, inductively, that a recursive call on edge set $\mathcal{E}_i$
eventually
contributes at most $|\mathcal{E}_i|\log|\mathcal{E}_i|/(6n^{\rho}\log m)$ edges
to $E_r$.  It follows from (C5) that the recursive call on edge set $E'$
contributes $|E'|\log|E'|/(6n^{\rho}\log m)$ edges to $E_r$.  We conclude
that $|E_r| \le |E|\log|E|/(6n^{\rho}\log|E|)=|E|/6n^{\rho}$.

Now we analyze the round complexity. In one recursive call of $\mathcal{A}^\ast$, consider a component $\mathcal{G}_i$ in the output partition, and let $K$ be the running time of vertices in $\mathcal{V}_i$.
Due to (C3), there are two cases. If $\mathcal{G}_i$ satisfied (C3-1),
it will halt in $K = O(n^{10\rho})$ rounds.
Otherwise, (C3-2) is met, and we have $|\mathcal{V}_i|\leq n' -\tilde{\Omega}(n^{\gamma})$.
Let $v \in V$ be any vertex, and let $K_1, \ldots, K_z$ be the running times of all calls to
$\mathcal{A}^\ast$ that involve $v$.
(Whenever $v$ ends up in $S$ or in a component satisfying (C3-1) it halts permanently, so $K_1,\ldots,K_{z-1}$ reflect executions that satisfy (C3-2) upon termination). Here $z$ can be at most $n^{1-\gamma}$, thus we have $\sum_{i=1}^z K_i \leq \tilde{O}(n^{1-\gamma  + 10\rho})$. And this is the running time since the whole algorithm stops within $\tilde{O}(n^{1-\gamma+9\rho})$ rounds.
\end{proof}

\subsection{Subroutines}
Before proving \cref{le:auxiliary}, we first introduce some helpful subroutines.
\cref{le:highdiameter} shows that for subgraphs of sufficiently high strong diameter, we can find a sparse cut of the subgraph, with runtime proportional to the strong diameter.
\cref{le:lowdegree} offers a procedure that removes a set of edges in such a way that the vertices in the remaining graph have high degree, and the removed edges form a low arboricity subgraph.
\cref{le:lowconductance} shows that if a subgraph already has a low conductance cut, then we can efficiently
find a cut of similar quality.

All these subroutines are applied to a connected subgraph $G^\ast=(V^\ast,E^\ast)$ of the underlying network $G=(V,E)$, and the computation does not involve vertices outside of $G^\ast$. In subsequent discussion in this section, the parameters $n$ and $m$ are always defined as $n = |V|$ and $m = |E|$, which are independent of the chosen subgraph  $G^\ast$.

\begin{lemma}\label{le:largediameter}
	Let $m$ and $D$ be two numbers.
	Let $(a_1,\ldots,a_D)$ be a sequence of positive integers such that  $D\geq 48n^{\rho}\log^2 m$ and $\sum_{i=1}^D a_i\leq m$.
	Then there exists an index $j$ such that $j\in [D/4, 3D/4]$
	and
	\[
	a_j\leq \frac{1}{12n^{\rho}\log m}\cdot \min\left(\sum_{i=1}^{j-1} a_i, \;\, \sum_{i=j+1}^D a_i\right).
	\]
\end{lemma}
\begin{proof} Define $S_k=\sum_{i=1}^k a_i$ to be the $k$th prefix sum. By symmetry, we may assume $S_{\lfloor D/2 \rfloor}\leq S_D-S_{\lfloor D/2 \rfloor}$, since otherwise we can reverse the sequence.
	Scan each index $j$ from $D/4$ to $D/2$. If an index $j$ does not satisfy $a_j\leq \frac{1}{12 n^{\rho}\log m} \cdot S_{j-1}$, then this implies that $S_j > S_{j-1}\left(1+\frac{1}{12n^{\rho}\log m}\right)$. If no index $j\in[D/4,D/2]$ satisfies this condition then $S_{\lfloor D/2 \rfloor}$ is larger than
	\[
	S_{\lfloor D/4 \rfloor} \cdot \left(1+\frac{1}{12n^{\rho}\log m}\right)^{D/4}
	\ge
	S_{\lfloor D/4 \rfloor} \cdot \left(1+\frac{1}{12n^{\rho}\log m}\right)^{12n^{\rho}\log^2 m}
	\geq S_{\lfloor D/4 \rfloor}\cdot m,
	\]
	which is impossible since $\sum_{i=1}^D a_i\leq m$. Therefore, there must exist an index $j\in[D/4,D/2]$ such that $a_j\leq \frac{1}{12n^{\rho} \log m}\cdot S_{j-1} = \frac{1}{12n^{\rho} \log m}\cdot \sum_{i=1}^{j-1}a_i$.
	By our assumption that $S_{\lfloor D/2 \rfloor}\leq S_D-S_{\lfloor D/2 \rfloor}$, we also have
	$a_j\leq \frac{1}{12n^{\rho} \log m}\cdot \min\left(\sum_{i=1}^{j-1}a_i,\;
	\sum_{i=j+1}^D a_i\right)$.
\end{proof}

\begin{lemma}[High Diameter subroutine]\label{le:highdiameter}
	Let $G^\ast=(V^\ast,E^\ast)$ be a connected subgraph and
	$x\in V^\ast$ be a vertex for which
	$\tilde{D} = \max_{v\in V^\ast} {\operatorname{dist}}_{G^\ast}(x,v) \geq 48n^{\rho}\log^2 m$.
	Define $\vlow=\{v\in V^\ast \mid \degree_{G^\star}(v)\leq{n^{\gamma}}/2\}$.
	Suppose there are no edges connecting two vertices in $\vlow$. Then we can find a cut $(C,\bar{C})$ of $G^\ast$ such that $\min(|C|,|\bar{C}|)\geq\frac{\tilde{D}}{32}n^{\gamma}$ and $\partial(C)\leq \min(\newvol(C),\newvol(\bar{C}))/ (12n^{\rho}\log m)$ in $O(\tilde{D})$ rounds deterministically in the $\CONGEST$ model. Each vertex in $V^\ast$ knows whether or not it is in $C$.
\end{lemma}
\begin{proof}
	The algorithm is as follows.
	First, build a BFS tree of  $G^\ast$ rooted at
	$x \in V^\ast$ in $O(\tilde{D})$ rounds.
	Let $L_i$ be the set of vertices of
	level $i$ in the BFS tree,
	and let $p_i$ be the number of edges $e=\{u,v\}$ such that $u\in L_i$ and $v\in L_{i+1}$.
	We write
	$L_{a..b} = \bigcup_{i=a}^b L_i$.
	In $O(\tilde{D})$ rounds we can let the root
	$x$ learn the sequence $(p_1, \ldots, p_{\tilde{D}})$.
	
	Note that in a BFS tree, edges do not connect
	two vertices in non-adjacent levels.
	By \cref{le:largediameter}, there exists an index $j\in[\tilde{D}/4,3\tilde{D}/4]$ such that $p_j\leq \frac{1}{12n^{\rho}\log m}\cdot \min\left(\newvol(L_{1..j}),\newvol(L_{j+1..\tilde{D}})\right)$,
	and such an index $j$ can be computed locally at the vertex $x$.
	
	The cut is chosen to be $C=L_{1..j}$, so we have $\partial(C)\leq \min(\newvol(C),\newvol(\bar{C}))/ (12n^{\rho}\log m)$. As for the second condition, due to our assumption in the statement of the lemma, for any two adjacent levels $L_i,L_{i+1}$, there must exist a vertex $v\in L_i\cup L_{i+1}$ such that $v\notin \vlow$.
	By definition of $\vlow$, $v$ has more than ${n^{\gamma}}/2$ neighbors in $G^\ast$, and they are all within $L_{i-1..i+2}$. Thus, the number of vertices within any four consecutive levels must be greater than ${n^{\gamma}}/2$. Since $j\in[\tilde{D}/4,3\tilde{D}/4]$, we have
	\[
	\min(|C|,|\bar{C}|)\geq \frac{\tilde{D}}{4}/4\cdot {n^{\gamma}}/2\geq \frac{\tilde{D}}{32}n^{\gamma}.
	\]
	To let each vertex in $V^\ast$ learn whether or not it is in $C$,  the root $x$ broadcasts the index $j$ to all vertices in $G^\ast$. After that, each vertex in level smaller than or equal to $j$ knows that it is in $C$; otherwise it is in $\bar{C}$.
\end{proof}

Intuitively, \cref{le:lowdegree} says that after the removal of a subgraph of small arboricity (i.e., the edge set $E_s^\diamond$), the remaining graph (i.e., the edge set $E^\diamond$) has high minimum degree. The runtime is proportional to the number of removed vertices (i.e., $|V^\ast|-|V^\diamond|$) divided by the threshold ${n^{\gamma}}$.
Note that the second condition of \cref{le:lowdegree} implies that $E_{s,v}^\diamond = \emptyset$ for all $v \in V^\diamond$.

\begin{lemma}[Low Degree subroutine]\label{le:lowdegree}
	Let $G^\ast=(V^\ast,E^\ast)$ be a connected subgraph with strong diameter $D$. We can partition $E^\ast =E^\diamond \cup E_s^\diamond$ meeting the following two conditions.
	\begin{itemize}
		\item[1.] Let $V^\diamond$ be the set of vertices induced by $E^\diamond$. Each $v \in V^\diamond$ has more than ${n^{\gamma}}/2$ incident edges in $E^\diamond$.
		\item[2.] The edge set $E_s^\diamond$ is further partitioned as $E_s^\diamond=\bigcup_{v\in V^\ast \setminus V^\diamond} E_{s,v}^\diamond$, where $E_{s,v}^\diamond$ is a subset of incident edges of $v$, and $|E_{s,v}^\diamond|\leq {n^{\gamma}}$. Each vertex $v$ knows $E_{s,v}^\diamond$.
	\end{itemize}
	This partition can be found in $O(D+(|V^\ast|-|V^\diamond|)/{n^{\gamma}})$  rounds deterministically in the $\CONGEST$ model.
\end{lemma}

\begin{proof}
	To meet Condition $1$, a naive approach is to iteratively ``peel off''  vertices that have degree at most ${n^{\gamma}}/2$, i.e., put all their incident edges in $E_s$, so long as any such vertex exists.
	On some graphs this process requires $\Omega(n)$ peeling iterations.
	
	We solve this issue by doing a batch deletion.
	First, build a BFS tree of $G^\ast$ rooted at an arbitrary vertex
	$x  \in V^\ast$.
	We use this BFS tree to let $x$  count the number of vertices that have degree less than ${n^{\gamma}}$ in the remaining subgraph in $O(D)$ rounds.
	
	The algorithm proceeds in iterations.
	Initially we set $E^\diamond \gets E^\ast$ and $E_s^\diamond \gets \emptyset$.
	In each iteration, we identify the subset $Z \subseteq V^\ast$
	whose vertices have at most ${n^{\gamma}}$ incident edges in $E^\diamond$.
	We orient all the $E^\diamond$-edges
	touching $Z$ away from $Z$, if one endpoint is in $Z$,
	or away from the endpoint with smaller ${{\operatorname{ID}}}$, if both endpoints are in $Z$.
	Edges incident to $v$ oriented away from $v$ are added to $E_{s,v}^\diamond$
	and removed from $E^\diamond$.
	The root $x$ then counts the number $z=|Z|$ of such vertices via the BFS tree.
	If  $z > {n^{\gamma}}/2$, we proceed to the next iteration; otherwise we terminate the algorithm.
	
	The termination condition ensures that each vertex has degree at least
	$({n^{\gamma}}+1) - z > {n^{\gamma}}/2$, and so Condition 1 is met.
	It is straightforward to see that
	the set $E_s^\diamond$ generated by the algorithm meets Condition 2,
	since for each $v$, we only add edges to $E_{s,v}^\diamond$ once,
	and it is guaranteed that $|E_{s,v}^\diamond| \leq  {n^{\gamma}}$.
	Tie-breaking according to vertex-${\operatorname{ID}}$ ensures the orientation is acyclic.

	Throughout the process, each time one vertex puts any edges into $E_s^\diamond$, it no longer stays in $V^\diamond$. Each iteration can be done in $O(D)$ time. We proceed to the next iteration only if there are more than  ${n^{\gamma}}/2$ vertices being removed from  $V^\diamond$. A trivial implementation can lead to an algorithm taking $O(D\left\lceil (|V^\ast|-|V^\diamond|)/{n^{\gamma}}\right\rceil))$ rounds.
	The round complexity can be further improved to
	$O(D+ (|V^\ast|-|V^\diamond|)/{n^{\gamma}})$ by pipelining the iterations. At some point the root $x$ detects that iteration $i$ was the last iteration; in $O(D)$ time it broadcasts a message to all nodes instructing them to
	roll back iterations $i+1,i+2,\ldots$, which have been executed speculatively.\end{proof}

The proof of the following lemma is given in \cite[Section 3]{chang2018distributed}

\begin{lemma}[Low Conductance subroutine]\label{le:lowconductance}
	Let $G^\ast=(V^\ast,E^\ast)$ be a connected subgraph with strong diameter $D$.
	Let $\phi\leq 1/12$ be a number.
	Suppose that there exists a subset $S\subset V^\ast$ satisfying
	\[
	\newvol(S)\leq (2/3)\newvol(V^\ast) ~~~\text{and}~~~ \Phi(S)\leq \frac{\phi^3}{19208\ln ^2(|E^\ast|e^4)}.
	\]
	Assuming such an $S$ exists,
	there is a $\CONGEST$ algorithm
	that finds a cut $C \subset V^\ast$
	such that $\Phi(C)\leq 12\phi$
	in $O(D + \poly (\log |E^\ast|, 1/\phi))$ rounds,
	with failure probability $1/\poly(|E^\ast|)$.
	Each vertex in $V^\ast$ knows whether or not it
	belongs to $C$.
\end{lemma}

\subsection{Proof of \cref{le:auxiliary}}

We prove \cref{le:auxiliary} by presenting and analyzing a specific distributed algorithm, which makes use of the subroutines specified
in \cref{le:highdiameter}, \cref{le:lowdegree}, and \cref{le:lowconductance}.

Recall that we are given a subgraph with edge set $E'$ and must ultimately
return a partition of it into $\EhExp'\cup E_s' \cup E_r'$.
The algorithm initializes $\EhExp' \gets E'$, $E_s' \gets \emptyset$,
and $E_r' \gets \emptyset$.
There are two types of special operations.
\begin{description}
	\item[Remove.] In an {\sf Remove} operation, some edges are moved from $\EhExp'$ to either $E_s'$ or $E_r'$. For the sake of a clearer presentation, each such operation is tagged \Remove{$i$}, for some index $i$.
	\item[Split.] Throughout the algorithm we maintain a partition of the current set $\EhExp'$. In a {\sf Split} operation, the partition subdivided. Each such operation is tagged as \Split{$i$}, for some index $i$, such that \Split{$i$} occurs right after \Remove{$i$}.
\end{description}

Throughout the algorithm, we ensure that any part $E^\star$ of the partition of $\EhExp'$ has an identifier that is known to all members of $V(E^\star)$.
It is not required that each part forms a connected subgraph.
The partition at the end of the algorithm,  $\EhExp'=\bigcup_{i=1}^t \mathcal{E}_i$, is the output partition.

\paragraph{Notations.}
Since we treat $\EhExp'$ as the ``active'' edge set and $E_s'$ and $E_r'$ as repositories of removed edges, $\deg(v)$ refers to the degree of $v$ in the subgraph induced by the \emph{current} $\EhExp'$.
We write $\vlow=\{v\in V' \mid \degree(v)\leq \ndel\}$.

\paragraph{Algorithm.}
In the first step of the algorithm,
move each edge $\{u,v\}\in \EhExp'$ in the subgraph
induced by $\vlow$to $E_{s,u}'$, where $\ID(u) < \ID(v)$ (\Remove{1}).
(Breaking ties by vertex-$\ID$ is critical to
keep the orientation acyclic.)

After that, $\EhExp'$ is divided into connected components. Assume these components are $G_1=(V_1,E_1)$, $G_2=(V_2,E_2), \ldots$, where $V_i = V(E_i)$.
Let $D_i$ be the depth of a BFS tree rooted at an arbitrary vertex in $G_i$.
In $O(D_i)$ rounds, the subgraph $G_i$ is assigned an identifier that is known to all vertices in $V_i$ (\Split{1}).
Note that this step is done in parallel for each $G_i$, and the time for this step is different for each $G_i$.
From now on there will be no communication between different subgraphs in $\{G_1, G_2, \ldots \}$, and we focus on one specific subgraph $G_i$ in the description of the algorithm.

Depending on how large $D_i$ is, there are two cases.
If $D_i \geq 48\log^2m$, we go to Case 1,
otherwise we go to Case 2.

\paragraph{Case 1:} In this case, we have
$D_i \geq 48n^{\rho}\log^2 m$.
Since there are no edges connecting two vertices in $\vlow$, we can apply the High Diameter subroutine, \cref{le:highdiameter}, which finds a cut $(C,\bar{C})$ of $G_i$ such that $\min(|C|,|V_i \setminus C|)\geq\frac{D_i}{32}n^{\gamma}$ and $\partial(C)\leq \min(\newvol(C),\newvol(V_i \setminus C))/ (12n^{\rho}\log m)$ in $O(D_i)$ rounds.
Every vertex in $V_i$ knows whether it is in $C$ or not. All edges of the cut $(C,\bar{C})$ are put into $E_r'$ (\Remove{2}). Then $E_i$ splits into two parts according to the cut  $(C,\bar{C})$ (\Split{2}).
After that, all vertices in $V_i$ terminate.
(Observe that the part containing the BFS tree root is connected, but the other part is not necessarily connected.)

\paragraph{Case 2:} In this case,
we have $D_i \leq 48n^{\rho} \log^2 m$.
Since $G_i=(V_i,E_i)$ is a small diameter graph,
a vertex $v \in V_i$ is able broadcast a message to all
vertices in $V_i$ very fast.
We apply the Low Degree subroutine, \cref{le:lowdegree}, to obtain a partition $E_i=E^\diamond \cup {E}_s^\diamond$.
We add all edges in ${E}_s^\diamond$ to $E_s'$ in such a way that $E_{s,v}' \gets E_{s,v}' \cup {E}_{s,v}^\diamond$ for all $v \in V_i \setminus V^\diamond$, where $V^\diamond = V(E^\diamond)$  (\Remove{3}).

After removing these edges, the remaining edges of $E_i$ are divided into several connected components, but all remaining vertices have degree larger than $\ndel/2$.
Assume these connected components are $G_{i,1}=(V_{i,1},E_{i,1})$, $G_{i,2}=(V_{i,2},E_{i,2})$, $\ldots$.
Let $D_{i,j}$ be the depth of the BFS tree from an arbitrary root vertex in $G_{i,j}$.
In $O(D_{i,j})$ rounds we compute such a BFS tree
and assign an identifier that is known to all vertices in
$V_{i,j}$ (\Split{3}).
That is, the remaining edges in $E_i$ are partitioned into $E_{i,1}$, $E_{1,2}$, $\ldots$.

In what follows, we focus on one subgraph $G_{i,j}$
and proceed to Case 2-a or Case 2-b.

\paragraph{Case 2-a:} In this case, $D_{i,j}\geq 48n^{\rho}\log^2 m$.
The input specification of the High Diameter subroutine (\cref{le:highdiameter}) is satisfied, since every vertex has degree larger than $\ndel/2$. We apply the High Diameter subroutine to $G_{i,j}$. This takes $O(D_{i,j})$ rounds. This case is similar to Case 1, and we do the same thing as what we do in Case 1, i.e., remove the edges in the cut found by the subroutine (\Remove{4}),  split the remaining edges (\Split{4}), and then all vertices in $V_{i,j}$ terminate.

\paragraph{Case 2-b:} In this case,
$D_{i,j} \leq 48n^{\rho} \log^2 m$.
Note that every vertex has degree larger than $\ndel/2$, and $G_{i,j}$ has small diameter. What we do in this case is to test whether $G_{i,j}$ has any low conductance cut; if yes, we will split $E_{i,j}$ into two components. To do so, we apply the Low Conductance subroutine, \cref{le:lowconductance}, with $\phi=\frac{1}{144n^{\rho}\log m}$.
Based on the result, there are two cases.

\paragraph{Case 2-b-i:} The subroutine finds a set of vertices $C$ that $\Phi(C)\leq 12\phi=\frac{1}{12n^{\rho}\log m}$, and every vertex knows whether it is in $C$ or not.
We move $\partial(C)$ to $E_r'$ (\Remove{5}), and then split the remaining edges into two edge sets according to the cut $(C, \bar{C})$ (\Split{5}). After that, all vertices in $V_{i,j}$ terminate.

\paragraph{Case 2-b-ii:} Otherwise, the subroutine does not return a subset $C$, and it means with probability at least $1 - 1/\poly(|E_{i,j}|) = 1 - 1/\poly(n)$, there is no cut $(S,\bar{S})$ with conductance less than $\frac{\phi^3}{19208\ln^2( |E_{i,j}| e^4)} = \Theta(\log^{-5}m)$.
Recall the relation between the mixing time
$\mix(G_{i,j})$ and the conductance $\Phi=\Phi_{G_{i,j}}$: $\Theta(\frac{1}{\Phi}) \leq \mix(G_{i,j}) \leq \Theta(\frac{\log |V_{i,j}|}{\Phi^2})$.
Therefore, w.h.p., $G_{i,j}$ has $O(\poly \log n)$ mixing time. All vertices in $V_{i,j}$ terminate without doing anything in this step.

Note that in the above calculation, we use the fact that every vertex in $V_{i,j}$ has degree larger than $\ndel/2$ in $G_{i,j}$, and this implies that $|V_{i,j}| = \Omega(\ndel)$ and $|E_{i,j}| = \Omega(\ndell)$, and so $\Theta(\log m) = \Theta(\log n) = \Theta(\log |E_{i,j}|) = \Theta(\log |V_{i,j}|)$.

\paragraph{Analysis.}
We show that the output of
$\mathcal{A}^\ast$ meets its
specifications (C1)--(C6).
Recall that
$\EhExp'=\bigcup_{i=1}^t \mathcal{E}_i$ is the final partition of the edge set $\EhExp'$ when all vertices terminate. Once an edge is moved from $\EhExp'$
to either $E_r'$ or $E_s'$,
it remains there for the rest of the computation.
Condition (C1) follows from the fact that each time we do a split operation, the induced vertex set of each part is disjoint. Condition (C6) follows from the fact that each vertex knows which part of $\EhExp'$ it belongs to after each split operation. In the rest of this section, we prove that the remaining conditions are  met.

\begin{claim}
	Condition (C2) is met.
\end{claim}
\begin{proof}
	Note that only \Remove{1} and \Remove{3} involve $E_s'$.
	In \Remove{1}, any $E_{s,u}'$ that becomes non-empty must have had
	$u\in \vlow$, so $\deg(u)\le \ndel$ before \Remove{1}, and
	therefore $|E_{s,u}'| + \deg(u) \leq \ndel$ after \Remove{1}.
	In \Remove{3}, the Low Degree subroutine of \cref{le:lowdegree} computes a partition $E_i=E^\diamond \cup {E}_s^\diamond$,
	and then we update $E_{s,u}' \gets E_{s,u}' \cup {E}_{s,u}^\diamond$ for all $u \in V_i \setminus V^\diamond$.
	By \cref{le:lowdegree}, for any $u$ such that ${E}_{s,u}^\diamond \neq \emptyset$, we have $|{E}_{s,u}^\diamond|\leq \ndel$, and $u \notin V^\diamond$, where $V^\diamond$ is the vertex set induced by the remaining edge set $E^\diamond$. In other words, once $u$ puts at least one edge  into $E_{s,u}'$, we have $\deg(u)=0$ after \Remove{3}.
\end{proof}

\begin{claim}
	Conditions (C3) and (C4) are met.
\end{claim}
\begin{proof}
	We need to verify that in each part of the algorithm, we either spend at most $\tilde{O}(n^{10\rho} )$ (because the run time of low conductance routine \cref{le:lowconductance} is $O(  \log^9/\phi^{10})$ rounds and here $\phi = \frac{1}{144n^{\rho\log m}}$, whereas the run time of high conductance cut is $O(n^{\rho}\log^2 m)$), or the size of the current component shrinks by $\tilde{\Omega}(n^{\gamma})$ vertices \emph{per round}.

	After removing all edges in the subgraph induced by $\vlow$, the rest of $E'$ is partitioned into
	connected components
	$\mathcal{E}_1,\mathcal{E}_2,\ldots$.
	Consider one such component $\mathcal{E}_i$,
	and suppose it goes to Case 1.
	We find a sparse cut $(C,\bar{C})$, and
	moving $\partial(C)$ to $E_r'$ breaks
	$\mathcal{E}_i$ into
	$\mathcal{E}_i^1$ and $\mathcal{E}_i^2$.
	By \cref{le:highdiameter},
	we have  $\min(|C|,|\bar{C}|) \geq \frac{D_i}{32}\ndel$, so the size of both $V(\mathcal{E}_i^1)=C$ and
	$V(\mathcal{E}_i^2) = \bar{C}$ are at most $|V(\mathcal{E}_i)|-\frac{D_i}{32}\ndel
	\leq n' - \Omega(D_i)\ndel = n' - \Omega(n^{\rho+\gamma})$.
	Since the running time for each vertex in $V(\mathcal{E}_i^1)$ and $V(\mathcal{E}_i^2)$
	is $O(D_i)$, the condition (C3-2) is met.
	
	Now suppose that $\mathcal{E}_i$ goes to Case 2.
	Note that the total time spent before it reaches Case 2 is $O(D_i)= \poly \log n$.
	In Case 2 we execute the Low Degree subroutine of \cref{le:lowdegree}, and let
	the time spent in this subroutine be $\tau$.
	By \cref{le:lowdegree}, it is either the case that (i) $\tau = O(D_i)$
	or
	(ii) the remaining vertex set $V^\diamond$ satisfies $|V(E_i)| - |V^\diamond| = \Omega(\tau \ndel)$.
	In other words, if we spend too much time (i.e., $\omega(D_i)$) on this subroutine, we must lose $\Omega(\ndel)$ vertices per round.
	
	After that, $\mathcal{E}_i$ is split into
	$\mathcal{E}_{i,1}$, $\mathcal{E}_{i,2}$, $\ldots$.
	We consider the set $\mathcal{E}_{i,j}$.
	If $\mathcal{E}_{i,j}$ goes to Case 2-a,
	then the analysis is the same as that in Case 1,
	and so (C3-2) is met.
	
	Now suppose that $\mathcal{E}_{i,j}$ goes to Case 2-b. Note that the time spent during the Low Conductance subroutine of \cref{le:lowconductance} is
	$\tilde{O}(n^{10\rho})$. Suppose that a low conductance cut $(C,\bar{C})$ is found (Case 2-b-i). Since the cut has conductance less than $\frac{1}{12n^{\rho}\log m}$, by the fact that every vertex has degree higher than $\ndel/2$, we must have $\min(|C|, |\bar{C}|) = \Omega(\ndel)$. Assume $\mathcal{E}_{i,j} \setminus \partial(C)$ is split into $\mathcal{E}_{i,j}^1$ and $\mathcal{E}_{i,j}^2$.
	The size of both $V(\mathcal{E}_{i,j}^1)$ and $V(\mathcal{E}_{i,j}^2)$
	must be at most $|V(\mathcal{E}_{i,j})|- \Omega(\ndel)$.
	Thus, (C3-2) holds for both parts
	$\mathcal{E}_{i,j}^1$ and $\mathcal{E}_{i,j}^2$.
	
	Suppose that no cut $(C,\bar{C})$ is found (Case 2-b-ii). If the running time $K$ among vertices in $V_{i,j}$ is $\tilde{O}(n^{10\rho})$, then (C3-1) holds.
	Otherwise, we must have $|{V}_{i,j}|\leq n' -\tilde{\Omega}(K\ndel)$ due to the Low Degree subroutine, and so (C3-2) holds.
	
	Condition (C4) follows from the the above proof of (C3), since for each part of the algorithm, it is either the case that
	(i) this part takes $O(n^{10\rho})$ time,
	or
	(ii) the number of vertices in the current subgraph is reduced by $\tilde{\Omega}(\ndel)$ \emph{per round}.
\end{proof}

\begin{claim}
	Condition (C5) is met.
\end{claim}
\begin{proof}
	Condition (C5) says that after the algorithm $\mathcal{A}^\ast$ completes, $|E_r'|\leq f$, where
	\[
	f=\left(|E'| \log|E'| -\sum_{i=1}^t |\mathcal{E}_i| \log |\mathcal{E}_i|\right)/ (6n^{\rho}\log m).
	\]
	We prove the stronger claim that this inequality holds at all times
	w.r.t.~the current edge partition
	$\mathcal{E}_1\cup \cdots\cup\mathcal{E}_t$ of $\EhExp'$.
	In the base case this is clearly true, since $t=1$ and
	$E' = \EhExp' = \mathcal{E}_1$ and
	$E_r'=\emptyset$.  Moving edges from $\EhExp'$ to $E_s'$
	increases $f$ and has no effect on $E_r'$, so we only have to consider the movement of edges from $\EhExp'$ to $E_r'$.
	Note that this only occurs in \Remove{$i$} and \Split{$i$}, for $i \in \{2,4,5\}$, where in these operations we find a cut
	$(C,\bar{C})$ and split one of the parts
	$\mathcal{E}_j$ according to the cut.
	In all cases we have
	\[
	|\partial(C)| \leq \frac{\min(\newvol(C),\newvol(\bar{C}))}{12n^{\rho}\log m}.
	\]
	Suppose that removing $\partial(C)$ splits $\mathcal{E}_j$ into $\mathcal{E}_j^1$ and $\mathcal{E}_j^2$,
	with $|\mathcal{E}_j^1| \le |\mathcal{E}_j^2|$
	and $C=V(\mathcal{E}_j^1)$.
	We bound the change in $|E_r'|$ and $f$ separately.  Clearly
	\begin{align*}
	\Delta|E_r'| &= |\partial(C)| \leq \frac{2|\mathcal{E}_j^1| + \partial(C)}{12n^{\rho}\log m} \leq \frac{|\mathcal{E}_j^1|}{6n^{\rho}\log m} + \frac{\partial(C)}{12n^{\rho}\log m}.
	\intertext{and}
	\Delta f &= \frac{1}{6n^{\rho}\log m}\cdot \left(|\mathcal{E}_j|\log|\mathcal{E}_j| - \sum_{k\in\{1,2\}} |\mathcal{E}_j^k|\log|\mathcal{E}_j^k|\right)\\
	&\ge \frac{1}{6n^{\rho}\log m}\cdot \left( |\mathcal{E}_j^1|\log(|\mathcal{E}_j|/|\mathcal{E}_j^1|) + \partial(C)\log|\mathcal{E}_j|\right)\\
	&> \Delta|E_r'|  &  \mbox{(Because $|\mathcal{E}_j^1| < |\mathcal{E}_j|/2$.)}
	\end{align*}
	Thus, $|E_r'| \leq f$ also holds after \Remove{$i$} and \Split{$i$},
	for $i\in\{2,4,5\}$.
\end{proof}